\documentclass[12pt]{article}
\usepackage{graphicx}
\usepackage{amsmath}
\usepackage{amssymb}
\usepackage{amsthm}
\usepackage{mathrsfs}
\newtheorem{dfn}{Definition}[section]

\newtheorem{prop}[dfn]{Proposition}
\newtheorem{lem}[dfn]{Lemma}

\newtheorem{rem}[dfn]{Remark}

\newtheorem{ex}[dfn]{Example}

\numberwithin{equation}{section}

\begin{document}

\title{Quantum walks on simplicial complexes}

\author{Kaname Matsue\thanks{The Institute of Statistical Mathematics, Tachikawa, Tokyo, 190-8562, Japan} $^{,}$\footnote{(Corresponding author) \tt kmatsue@ism.ac.jp} ,\ Osamu Ogurisu\thanks{Division of Mathematical and Physical Sciences,
  Kanazawa University, Kanazawa, Ishikawa 920-1192, Japan} $^{,}$\footnote{\tt ogurisu@staff.kanazawa-u.ac.jp}$\ $ and Etsuo Segawa\thanks{Graduate school of Information Sciences, Tohoku University, Aoba, Sendai, 980-8579, Japan}  $^{,}$\footnote{\tt e-segawa@m.tohoku.ac.jp}}
\maketitle
\begin{abstract}
We construct a new type of quantum walks on simplicial complexes as a natural extension of the well-known Szegedy walk on graphs.
One can numerically observe that our proposing quantum walks possess linear spreading and localization
as in the case of the Grover walk on lattices.
Moreover, our numerical simulation suggests that localization of our quantum walks reflect not only topological but also geometric structures.
On the other hand, our proposing quantum walk contains an intrinsic problem concerning exhibition of nontrivial behavior, which is not seen in typical quantum walks such as Grover walks on graphs.
\end{abstract}


\section{Introduction -- Motivation and Aim --}
\label{section-intro}
The quantum walk is a quantum analogue of classical random walks \cite{Gu}. 
Its primitive form of the discrete-time quantum walk on $\mathbb{Z}$ can be seen in Feynman's checker board \cite{FH}. 
It is mathematically shown (e.g. \cite{K1}) that this quantum walk has a completely different limiting behavior from classical random walks, which is a typical example showing a  difficulty of intuitive description of quantum walks' behavior. 
By such an interesting observation and the efficiency of quantum walks in quantum search algorithms (see \cite{A1, K2} and their references), quantum walks are studied from various kinds of viewpoints such as the quantum information \cite{A2, AKR, SKW}, approximation of physical process derived from the Dirac and Schr\"{o}dinger equations \cite{CBS, S}, experimental and industrial viewpoints \cite{KFCSAMW, MY, ZKGSBR, MW, BBW}, and so on.  

A primitive form of the Grover walk on graphs has been appeared in \cite{W}. This is considered as one of the most intensively-investigated quantum walks from the viewpoint of quantum information theory, a spectral graph theory \cite{A1, HKSS1} and a scattering theory \cite{HKSS1}. 
As a generalization of the Grover walk, the Szegedy walk is proposed to provide more abstractive discussions of applications such as quantum search algorithms \cite{Sze}. 
Very recently, the Szegedy walk is applied to quantum PageRank algorithm on large complex networks \cite{PM, PMCM}, for example.
A key feature of the Szegedy walk is that the spectrum is decomposed into two parts : the inherited part associated with underlying random walks on graphs and the birth part with some multiplicities. 
In particular, the birth part reflects the cycle structure on graphs, and it induces localization of quantum walks, namely, 
the finding probability of quantum walks remains as a positive value even in the long time limit. 
These features are also discussed in \cite{HKSS2} in the case of the Grover walk on crystal lattices.

\bigskip
The main aim of this paper is to construct a geometric multi-dimensional analogue of quantum walks on graphs.
In earlier works, quantum walks on geometric objects embedded in multi-dimensional spaces are considered as ones on {\em graphs}.
Such a traditional quantum walk can be regarded as transmission and reflection at vertices of one dimensional waves on edges of the graph.
On the other hand, if we consider an analogue of quantum walks related to multi-dimensional waves, it is natural to consider quantum walks on geometric objects which admits multi-dimensional structures like planes or surfaces. 
Once one can construct such quantum walks, it is expected to study multi-dimensional feature of walks as well as their deeper geometric aspects.

In this paper we take {\em simplicial complexes} as a natural multi-dimensional extensions of graphs.
The essence of the construction of quantum walks on simplicial complexes 
is to make an interaction of multi-dimensional waves satisfying the following postulates 1-3 in quantum mechanics \cite{NC};
\begin{enumerate}
\item  Underlying space of dynamical system is a Hilbert space $\mathcal{H}$;　
\item  The time evolution is governed by a unitary operator on $\mathcal{H}$;  
\item  There exists a collection of orthogonal projections\footnote{We restrict ourselves to orthogonal projections for the measurement for simplicity.}
 $\{ E_m \}$ with $\sum_m E_m=I$ such that if the state $\psi\in \mathcal{H}$ is measured, then the probability that the result $m$ occurs is given by
\[ Prob(m)=\|E_m \psi\|_\mathcal{H}^2. \]
\item Moreover we provide the following additional postulate (cf. No-go lemma \cite{Meyer}): \\ 
{\it the time evolution can provide a non-trivial interaction,} in particular, transmission and reflection of waves.
\end{enumerate}
In this paper, we propose a notion of the non-trivial interaction as follows: we say non-trivial interaction exhibits if and only if the time evolution $U$ is neither {\em tethered} (Definition \ref{dfn-tethered}) nor {\em non-interactive} (Definition \ref{dfn-non-interactive}, cf. No-go lemma \cite{Meyer}). 

For given $n$-dimensional simplicial complex $\mathcal{K}=\{K_k\}_{k=0}^n$, 
we propose the following walk 
which meets at least postulates 1--3 (see Definition \ref{dfn-S-QW} for details)
and we check that the postulate 4 obviously holds for the simplicial complexes treated here by numerical simulations.
\begin{enumerate}
\item {\bf Total state space }(postulate 1): 
We define a \lq\lq directed" $n$-simplex by 
  \[ \tilde{K}_n=\{ \pi[a_0 a_1\cdots a_n] \mid  |a_0a_1\cdots a_n| \in K_n, \pi\in \mathcal{S}_{n+1} \}. \]
The total state space is $\mathcal{H}=\ell^2(\tilde{K}_n)$ with the standard inner product 
(see Definition \ref{dfn-class-complex} for the detailed setting of the simplicial complex).
\item  {\bf Time evolution }(postulates 2): 
Let $d_0^{(n)}: \ell^2(\tilde{K}_{n})\to \ell^2(\tilde{K}_{n-1})$ be a coisometric operator whose explicit expression is denoted by Definition \ref{dfn-dn}, 
and also let $S: \ell^2(\tilde{K}_{n})\to \ell^2(\tilde{K}_{n})$ be a cyclic shift operator. 
Then the one-step time evolution operator is  defined by 
  \[ U=S(2{d_0^{(n)}}^*d_0^{(n)}-I), \]
where ${d_0^{(n)}}^\ast$ is the adjoint operator of $d_0^{(n)}$. 
This is a unitary operator on $\ell^2(\tilde{K}_{n})$ (Proposition \ref{unitary-C}). 
After one step unitary map $U$, an incident state in a simplex is changed to a linear combination of the states 
in the same simplex (reflection) and its adjacent simplices (transmission). See Figure 1. 
\item {\bf Finding probability }(postulate 3): 
We define the finding probability at a simplex $|\sigma|\in K_n$ with the initial state $\psi_0\in \ell^2(\tilde{K}_{n})$ after $m$-th iteration is 
  \[ \mu_m^{\psi_0} (|\sigma|)= \sum_{\pi\in \mathcal{S}_{n+1}}  (U^m\psi_0)(\pi\sigma) ^2_{\tilde{K}_{n}}.  \]
\end{enumerate}
Remark that, when $n=1$, the original Szegedy walk on a graph \cite{Sze, HKSS2} is reconstructed. 
In that sense, this walk is a natural generalization of the Szegedy walk. 
In the above settings, we run this walk on several kinds of simplicial complexes by numerical simulations in this paper.  
Our numerical simulation results provide the following interesting observations.

\par
Quantum walks on simplicial complexes reflect topology of simplicial complexes, as in the case of lattices discussed in \cite{HKSS2}. Moreover, multi-dimensional geometric features yield richer aspects of our quantum walks.
More precisely, we numerically obtain the following observations:
\begin{itemize}
\item Our quantum walks admit linear spreading and localizations, which are similar behaviors to traditional quantum walks on graphs (Figures \ref{fig-spread} and \ref{fig-time-average}). 
\item Topological cycle structure brings about localization of quantum walks (Figures \ref{fig-Swalk-cylinder} and \ref{fig-time-average}). 
\item Higher dimensional topological cycles such as cavities (Figure \ref{fig-tetrahedron}) can \lq\lq absorb" localized states on cycles (Figures \ref{fig-tetrahedron}, \ref{fig-Swalk-cylinder} and \ref{fig-time-average}). 
\item Orientability of simplicial complexes also affects the behavior of quantum walks (Figures \ref{fig-Mobius} and \ref{fig-time-average}).
\end{itemize}

\bigskip
This paper is organized as follows. 
In Section \ref{section-construction}, we construct a unitary operator on simplicial complexes which is a natural extension of quantum walks on graphs. 
In Section \ref{section-numerical}, we numerically study quantum walks constructed in Section \ref{section-construction}. 
In particular, we study our quantum walks on the following geometric objects:
\begin{itemize}
\item $\mathbb{R}^2$: two dimensional Euclidian space corresponding to the simplicial complex $\mathcal{K}_0$;
\item $S^1\times \mathbb{R}$: a infinite cylinder corresponding to the simplicial complex $\mathcal{K}_1$;
\item an infinite cylinder with a tetrahedron with cavity corresponding to the simplicial complex $\mathcal{K}_2$;
\item the M\"{o}bius band corresponding to the simplicial complex $\mathcal{K}_3$.
\end{itemize}
We observe the similar and different properties compared with the traditional quantum walks on graphs as stated in the above.
Concrete implementations of our quantum walks and numerical simulation results are shown in Appendix. 

\section{Quantum walks on simplicial complexes}
\label{section-construction}

Let $\mathcal{K}$ be an $n$-dimensional simplicial complex and $K_k = K_k(\mathcal{K})$ be a collection of $k$-simplices 
which belongs to $\mathcal{K}$. 
Throughout this paper we consider the following class of simplicial complexes.
\begin{dfn}\rm
\label{dfn-class-complex}
We shall call an $n$-dimensional simplicial complex $\mathcal{K} = \{K_k\}_{k=0}^n$ {\em admissible} if the following conditions hold:
\begin{itemize}
\item $\mathcal{K}$ is strongly connected. See also Definition \ref{dfn-strong-conn};
\item For each $k=0,\cdots, n-1$, every $|\tau| \in K_k$ is a primary face of some $|\sigma| \in K_{k+1}$. We do not assume that such a $(k+1)$-simplex $|\sigma|$ is uniquely determined. 
Moreover, assume that there is a positive integer $M$ such that
\begin{equation*}
\sharp\{|\sigma| \in K_{k+1}\mid |\sigma| \text{ admits }|\tau| \text{ as a face}\} \leq M < \infty
\end{equation*}
\end{itemize}
holds for all $|\tau| \in K_k$ and for each $k=0,\cdots, n-1$.
\end{dfn}
In this paper, we define the set of \lq\lq directed" $k$-simplices $\tilde K_k\ (k=0,\cdots, n-1)$  associated with $K_k$ by the following, which is important to construct our quantum walks.
\begin{equation*}
\tilde K_k:= \{\pi \sigma:= \pi [a_0 a_1\cdots a_k] \mid |\sigma| = |a_0 a_1 \cdots a_k|\in K_k,\ \pi \in \mathcal{S}_{k+1}\},
\end{equation*}
where $\mathcal{S}_{k+1}$ is the $(k+1)$-dimensional permutation group. For example, a $2$-simplex $|abc| \in K_2$ generates six different elements in $\tilde K_2$: $[abc]$, $[bca]$, $[cab]$, $[acb]$, $[cba]$ and $[bac]$.
Whenever we use these notations, we distinguish $[abc]$ from $[bca]$, $[acb]$ and so on. 
On the other hand, we identify $|abc|$ with $|bca|$, $|abc|$, etc., which can be regarded as {\em the support} of directed simplices.

\bigskip
Now we define quantum walks on $\mathcal{K}$. Firstly, define a $\mathbb{C}$-linear space $\ell^2(\tilde K_k)$ by 
\begin{equation*}
\ell^2(\tilde K_k):= \left\{f : \tilde K_k\to \mathbb{C} \mid \| f \|_{\tilde K_k} < \infty\right\}.
\end{equation*}
Here the inner product is given by the standard inner product, that is,
\begin{equation*}
\langle f, g \rangle_{\tilde K_k}:= \sum_{\sigma \in \tilde K_k}\overline{f(\sigma)}g(\sigma).
\end{equation*}
Let $\|\cdot \|_{\tilde K_k}$ be the associated norm, namely, $\|f\|_{\tilde K_k} := \langle f, f \rangle_{\tilde K_k}^{1/2}$.
We take
\begin{equation*}
\delta^{(k)}_{\sigma}(\sigma'):= \begin{cases}
	1 & \text{ if $\sigma' = \sigma$}\\	
	0 & \text{ if $\sigma' \not = \sigma$}
\end{cases}
\end{equation*}
as the standard basis of $\ell^2(\tilde K_k)$. One knows that the $\mathbb{C}$-linear space $\ell^2(\tilde K_k)$ associated with the inner product $\langle \cdot, \cdot \rangle_{\tilde K_k}$ is a Hilbert space.

\bigskip
Fix a permutation on $(n+1)$-words $\pi \in \mathcal{S}_{n+1}$ whose order is $n+1$. That is,
\begin{equation*}
\pi^1 \not = \pi^2 \not = \cdots \not = \pi^{n+1} = I.
\end{equation*}
As an example, we choose 
\begin{equation}
\label{sample-permutation}
\pi [a_0 a_1 \cdots a_{n-1} a_n] = [a_1 a_2 \cdots a_n a_0].
\end{equation}
Throughout this paper, we only consider the permutation (\ref{sample-permutation}) for simplicity.
We associate this permutation with the $n$-th total state space $\ell^2(\tilde K_n)$ by
\begin{equation*}
S_\pi \delta^{(n)}_{\sigma}:= \delta^{(n)}_{\pi \sigma},\quad \sigma \in \tilde K_n.
\end{equation*}

\begin{dfn}[Total state space]\rm
We shall say the Hilbert space $(\ell^2(\tilde K_n), \langle \cdot, \cdot \rangle_{\tilde K_n})$ {\em the total state space} of the simplicial quantum walk on $\mathcal{K}$ defined below, where $n = \dim \mathcal{K}$.
\end{dfn}

Secondly, define a linear operator on the total state space and its adjoint operator.
\begin{dfn}\rm
\label{dfn-dn}
Let $i\in \{0,1,\cdots, n\}$. Define a map $\tilde d^{(n)}_i: \tilde K_n\to \tilde K_{n-1}$ by the following:
\begin{equation*}
\tilde d^{(n)}_i [a_0a_1\cdots a_n]:= [a_{i+1}a_{i+2}\cdots a_n a_0 a_1 \cdots a_{i-1}].
\end{equation*}
Also, define $d^{(n)}_i:\ell^2({\tilde K_n}) \to \ell^2({\tilde K_{n-1}})$ by the $\mathbb{C}$-linear extension of the following expression: 
\begin{equation*}
d^{(n)}_i \delta^{(n)}_{\sigma} = \overline{w(\pi^i\sigma)}\delta_{\tilde d^{(n)}_i\sigma}^{(n-1)},
\end{equation*}
where $w(\sigma) \in \mathbb{C}$ is a constant depending on $\sigma\in \tilde K_n$, 
and $\overline{w}$ is the complex conjugate of $w\in \mathbb{C}$. 
\end{dfn}

For example, consider the case where $n=2$ and $\sigma = [abc]$.
We then have
\begin{equation*}
\tilde d^{(2)}_0([abc]) = [bc],\quad \tilde d^{(2)}_1([abc]) = [ca],\quad \tilde d^{(2)}_2([abc]) = [ab],
\end{equation*}
which yield
\begin{equation*}
d^{(2)}_0\delta^{(2)}_{[abc]} = \overline{w([abc])} \delta^{(1)}_{[bc]},\quad d^{(2)}_1\delta^{(2)}_{[abc]} = \overline{w([bca])} \delta^{(1)}_{[ca]},\quad d^{(2)}_0\delta^{(2)}_{[abc]} = \overline{w([cab])} \delta^{(1)}_{[ab]}.
\end{equation*}
One easily sees that $d^{(n)}_i = d^{(n)}_0 \circ S_\pi^j$ holds for $i=0, \cdots, n$. 

\begin{dfn}\rm
Define a linear operator ${d^{(n)}_i}^\ast: \ell^2(\tilde K_{n-1})\to \ell^2(\tilde K_n)$ by the $\mathbb{C}$-linear extension of the following formulation:
\begin{equation*}
{d^{(n)}_i}^\ast \delta_{\tau}^{(n-1)}:= \sum_{\sigma: \tilde d^{(n)}_i\sigma = \tau}w(\pi^i\sigma)\delta_{\sigma}^{(n)}.
\end{equation*}
\end{dfn}

\begin{rem}\rm
Easy calculations yield
\begin{equation*}
S_\pi {d^{(n)}_1}^\ast = {d^{(n)}_0}^\ast,\quad S_\pi {d^{(n)}_2}^\ast = {d^{(n)}_1}^\ast,\quad \cdots,\quad S_\pi {d^{(n)}_n}^\ast = {d^{(n)}_{n-1}}^\ast,\quad S_\pi {d^{(n)}_0}^\ast = {d^{(n)}_n}^\ast,
\end{equation*}
since the unitary transpose $S_\pi^\ast$ of $S_\pi$ coincides with $S_\pi^{-1}$.
\end{rem}
The linear operator ${d^{(n)}_i}^\ast$ is indeed the adjoint operator of $d^{(n)}_i$ in the following sense.

\begin{lem}
\label{lem-adjoint}
Assume that the function $w: \tilde K_n \to \mathbb{C}$ in the definition of $d_i$ is bounded. 
Then, for each $i=0,\cdots, n$, the following equality holds: 
\begin{equation*}
\langle d^{(n)}_i \psi,\phi \rangle_{\tilde K_{n-1}} = \langle \psi,{d^{(n)}_i}^\ast \phi \rangle_{\tilde K_n},\quad \forall \psi\in \ell^2(\tilde K_n),\ \phi \in \ell^2(\tilde K_{n-1}).
\end{equation*}
\end{lem}

\begin{proof}
Let 
$\psi = \sum_{\sigma\in \tilde K_n}\psi_{\sigma} \delta_{\sigma}^{(n)} \in \ell^2(\tilde K_n)$ 
and $\phi = \sum_{\tau\in \tilde K_{n-1}}\phi_{\tau} \delta_{\tau}^{(n-1)} \in \ell^2(\tilde K_{n-1})$, where $\psi_{\sigma},\ \phi_{\tau} \in \mathbb{C}$.
Direct computations yield
\begin{align*}
\langle \psi,{d^{(n)}_i}^\ast \phi \rangle_{\tilde K_n} &= \sum_{\sigma \in \tilde K_n} \overline{\psi(\sigma)} ({d^{(n)}_i}^\ast \phi) (\sigma)\\
	&= \sum_{\sigma \in \tilde K_n}\overline{\psi(\sigma)}
	\left( \sum_{\tau \in \tilde K_{n-1}} \phi_{\tau} ({d^{(n)}_i}^\ast \delta_{\tau}^{(n-1)})(\sigma) \right)\\
	&= \sum_{\sigma \in \tilde K_n}\overline{\psi(\sigma)}
	\left( \sum_{\tau \in \tilde K_{n-1}} \phi_{\tau} \left(\sum_{\sigma' \in \tilde K_n: \tilde d^{(n)}_i {\sigma'} = \tau}w(\pi^i \sigma') \delta_{\sigma'}^{(n)}(\sigma) \right) \right)\\
	&= \sum_{\sigma \in \tilde K_n}\overline{\psi(\sigma)}
	\left( \phi_{\tilde d^{(n)}_i\sigma} w^{(n)}(\pi^i\sigma) \delta_{\sigma}^{(n)}(\sigma) \right)\\
	&= \sum_{\sigma \in \tilde K_n}\overline{\psi_{\sigma}} w(\pi^i\sigma) \phi_{\tilde d^{(n)}_i \sigma}.
\end{align*}
Similarly,
\begin{align*}
\langle d^{(n)}_i \psi,\phi\rangle_{\tilde K_{n-1}} &= \sum_{\tau \in \tilde K_{n-1}}\left( \sum_{\sigma\in \tilde K_n}w(\pi^i\sigma) \overline{\psi_{\sigma}} \delta_{\tilde d^{(n)}_i\sigma}^{(n-1)}(\tau) \right)\left(\sum_{\tau' \in \tilde K_{n-1}}\phi_{\tau'} \delta_{\tau'}^{(n-1)}(\tau)\right)\\
	&= \sum_{\tau \in \tilde K_{n-1}}\left( \sum_{\sigma\in \tilde K_n}w(\pi^i\sigma) \overline{\psi_{\sigma}} \delta_{\tilde d^{(n)}_i \sigma}^{(n-1)}(\tau)\right)\phi_{\tau}.
\end{align*}
Now the summand $\sum_{\sigma\in \tilde K_n}w(\pi^i\sigma) \overline{\psi_{\sigma}} \delta_{\tilde d^{(n)}_i \sigma}^{(n-1)}(\tau)$ is actually $\sum_{\sigma : \tilde d_i^{(n)}\sigma = \tau }w(\sigma) \overline{\psi_\sigma}$, which is a finite sum since $\mathcal{K}$ is admissible.
Since each $|\tau|\in K_{n-1}$ has a coface $|\sigma| \in K_n$, then the sum $\sum_{\tau \in \tilde K_{n-1}} \sum_{\sigma : \tilde d_i^{(n)}\sigma = \tau }$ is exactly the same as $\sum_{\sigma \in \tilde K_n}$. Therefore
\begin{align*}
\langle d^{(n)}_i \psi,\phi\rangle_{\tilde K_{n-1}}
		&= \sum_{\sigma \in \tilde K_n}\overline{\psi_{\sigma}} w(\pi^i\sigma) \phi_{\tilde d^{(n)}_i \sigma}.
\end{align*}
Finally, calculations in both $\langle \psi,{d^{(n)}_i}^\ast \phi \rangle_{\tilde K_n}$ and $\langle d^{(n)}_i \psi,\phi\rangle_{\tilde K_{n-1}}$ make sense since $w$ is bounded. As a result, the proof is completed.
\end{proof}

\begin{dfn}\rm
We shall say the function $w: \tilde K_n \to \mathbb{C}$ {\em a weight on $\mathcal{K}$} if $w(\sigma)\not = 0$ for all $\sigma \in \tilde K_n$ and
\begin{equation*}
\sum_{\sigma \in \tilde K_n: \tilde d^{(n)}_0 \sigma=\tau} |w(\pi^i\sigma)|^2 = 1
\end{equation*}
holds for all $\tau \in \tilde K_{n-1}$.
\end{dfn}
The second statement of the weight is equivalent to
\begin{equation*}
\sum_{\sigma \in \tilde K_n: \tilde d^{(n)}_i \sigma =\tau}|w(\sigma)|^2 = 1,\quad \forall \tau \in \tilde K_{n-1},\quad i=0,\cdots, n.
\end{equation*}
By the definition, a weight $w$ on an admissible simplicial complex $\mathcal{K}$ satisfies the assumption of Lemma \ref{lem-adjoint}. 
Using the notion of weights, we obtain the following proposition, which is the center for constructing our quantum walk.

\begin{prop}
\label{unitary-C}
Let $\mathcal{K}$ be an admissible $n$-dimensional simplicial complex which admits a weight $w$. Then, for $i=0,\cdots, n$, the linear operator $C_i:= 2{d^{(n)}_i}^\ast d^{(n)}_i - I$ is a unitary operator on $\ell^2(\tilde K_n)$.
\end{prop}

\begin{proof}
Admissibility of $\mathcal{K}$ and the existence of a weight $w$ yield that both $d_i^{(n)}$ and ${d_i^{(n)}}^\ast$ are bounded and linear.
Since $C_i^\ast = (2{d^{(n)}_i}^\ast d^{(n)}_i - I)^{\ast} = C_i$, it is sufficient to prove $C_i^2 =I$.
One knows 
\begin{align*}
C_i^2 &= (2{d^{(n)}_i}^\ast d^{(n)}_i - I)^2 = 4{d^{(n)}_i}^\ast(d^{(n)}_i {d^{(n)}_i}^\ast - I) d^{(n)}_i + I.
\end{align*}
Now we consider the linear operator $d_i d_i^\ast$. 
By the definition, we obtain 
\begin{align*}
d^{(n)}_i {d^{(n)}_i}^\ast \delta_{\tau}^{(n-1)} 
&= d^{(n)}_i\left(\sum_{\sigma \in \tilde K_n: \tilde d^{(n)}_i \sigma = \tau} w(\pi^i\sigma) \delta_{\sigma}^{(n)}\right)\\
& = \sum_{\sigma \in \tilde K_n: \tilde d^{(n)}_i \sigma = \tau} \overline{w(\pi^i\sigma)} \cdot w(\pi^i\sigma) \delta_{\tilde d^{(n)}_i \sigma}^{(n-1)}.
\end{align*}
The function $\delta_{\tilde d^{(n)}_i \sigma}^{(n-1)}$ in the sum of the right-hand side is equal to $\delta_{\tau}^{(n-1)}$. 
Since $w$ is a weight on $\mathcal{K}$, then we obtain
\begin{equation*}
\sum_{\sigma \in \tilde K_n: \tilde d_i \sigma  = \tau} \overline{w(\pi^i\sigma)} \cdot w(\pi^i\sigma) \delta_{\tau}^{(n-1)} = \delta_{\tau}^{(n-1)}.
\end{equation*}
Consequently, $d^{(n)}_i {d^{(n)}_i}^\ast \delta_{\tau}^{(n-1)}  = \delta_{\tau}^{(n-1)} $ holds. 
Since $\tau \in \tilde K_{n-1}$ is arbitrary, then $d^{(n)}_i {d^{(n)}_i}^\ast = I$ holds. 
The equality $C_i^2 = I$ thus holds and the proof is completed.
\end{proof}

\begin{dfn}[Simplicial quantum walks]\rm
\label{dfn-S-QW}
Let $\mathcal{K}$ be an admissible $n$-dimensional simplicial complex which admits a weight $w$. 
We shall say $C_i$ defined in Proposition \ref{unitary-C} a {\em (local) quantum coin on $\mathcal{K}$} and an operator $S_\pi$ a {\em shift operator} associated with a permutation $\pi\in \mathcal{S}_{n+1}$. 
An element $f\in \ell^2(\tilde K_n)$ with $\|f\|_{\tilde K_n} = 1$ is called {\em a state}.
For the unitary operator $U:= S_\pi C_0$ and a state $f\in \ell^2(\tilde K_n)$, define the (finding) {\em probability at time $m$ at $|\sigma| \in K_n$ with the initial state $f$} by
\begin{equation*}
\mu_m^{(f)} (|\sigma|) := \sum_{\tau : |\tau|=|\sigma| }\|U^m f(\tau )\|_{\tilde K_n}^2.
\end{equation*}

Finally, we call the triple $(U, \ell^2(\tilde K_n), \{\mu_m\}_{m\geq 0})$, or simply $U$, the {\em simplicial quantum walk}, or {\em S-quantum walk} for short, {\em on $\mathcal{K}$}.
\end{dfn}

Consider, as an example, the case $n=2$. The S-quantum walk $U$ scatters a state on a $2$-simplex $|\sigma|$ to the linear combination of a rotating state on the same $2$-simplex $|\sigma|$ and transmitted states to simplices adjacent to $|\sigma|$. 
The resulting state can be considered as the collection of {\em the reflection (the former) and the transmission (the latter) of two dimensional waves}. 
The illustration of the S-quantum walk on $\mathcal{K}$ is shown in Figure \ref{fig-Swalk}. 
It is a generalization of phenomena in quantum walks on graphs, as shown in Figure \ref{fig-1dimwalk}.

Note that 
$\sum_{|\sigma| \in K_n} \mu_m^{(f)}(|\sigma|)=1$ always holds for any state $f$ and all $m$, since $\|f\|_{\tilde K_n} = 1$ and $U$ is unitary. 

\begin{figure}[htbp]\em
\begin{minipage}{0.32\hsize}
\centering
\includegraphics[width=4.0cm]{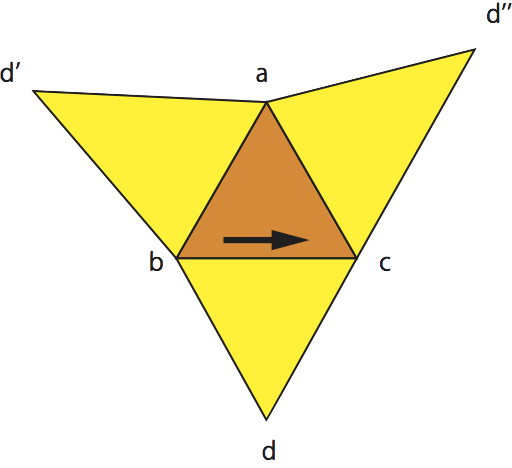}
(a)
\end{minipage}
\begin{minipage}{0.32\hsize}
\centering
\includegraphics[width=4.0cm]{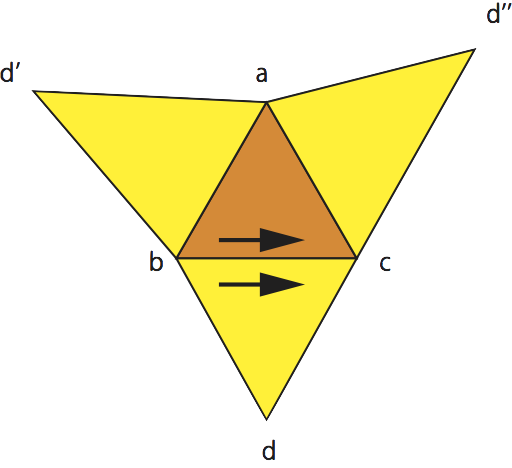}
(b)
\end{minipage}
\begin{minipage}{0.32\hsize}
\centering
\includegraphics[width=4.0cm]{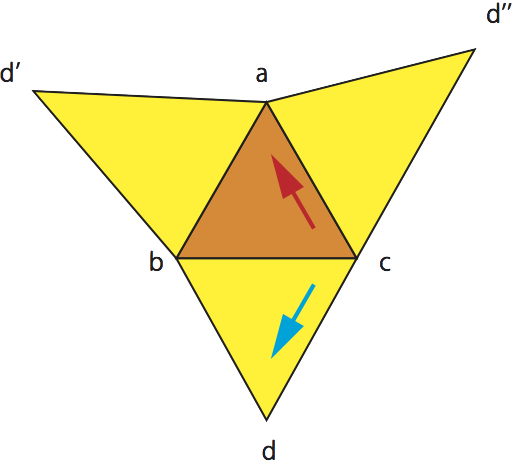}
(c)
\end{minipage}
\caption{S-quantum walk $U$.}
\label{fig-Swalk}
(a). The initial state $\delta^{(2)}_{[abc]}$.

(b). The quantum coin $C_0$. This figure shows the state $C_0\delta^{(2)}_{[abc]} = (2|w([abc])|^2-1)\delta^{(2)}_{[abc]} + 2\overline{w([abc])}w([dbc])\delta^{(2)}_{[dbc]}$.

(c). The S-quantum walk $U=S_\pi C_0$. This figure shows the state $U\delta^{(2)}_{[abc]} = (2|w([abc])|^2-1)\delta^{(2)}_{[bca]} + 2\overline{w([abc])}w([dbc])\delta^{(2)}_{[bcd]}$. The state on $|\sigma| = |abc|$ drawn by the red arrow represents the \lq\lq reflection" of the state on $|\sigma|$. Similarly, the state on $|bcd|$ drawn by the blue arrow represents the \lq\lq transmission" of the state on $|\sigma|$ to $|bcd|$.
\end{figure}

\begin{figure}[htbp]\em
\begin{minipage}{0.5\hsize}
\centering
\includegraphics[width=4.0cm]{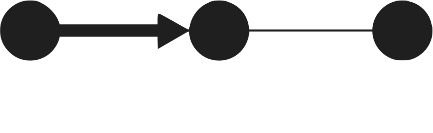}
(a)
\end{minipage}
\begin{minipage}{0.5\hsize}
\centering
\includegraphics[width=4.0cm]{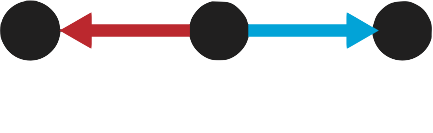}
(b)
\end{minipage}
\caption{Quantum walk on graph.}
\label{fig-1dimwalk}
(a). Initial state of $1$-dimensional quantum walk on an edge.

(b). Behavior of quantum walk on graph. Initial state (a) is mapped to the state via reflection (red) and transmission (blue). S-quantum walk on $\mathcal{K}$ described in Figure \ref{fig-Swalk} can be considered as the multi-dimensional generalization of this behavior.
\end{figure}

\begin{rem}\rm
Even if $\tau \in K_{n-1}$ is an $(n-1)$-simplex whose coface is only one $n$-simplex $\sigma$ and $\tilde d_i \sigma =\tau$, 
the identification of $C_i\delta_{\sigma}^{(n)}= \delta_{\sigma}^{(n)}$ also makes sense.
In this case, we can interpret that the S-quantum walk $U$ only represents the reflection of waves without transmission.
\end{rem}

\begin{rem}\rm
If $n=1$, 
the S-quantum walk on $\mathcal{K}$ is nothing but the Szegedy walk \cite{Sze}. Moreover, if we assume that $w(e) = 1/\sqrt{\deg(\tilde d_1^{(1)}e)}$ for all $e\in \tilde K_1$, 
then the corresponding S-quantum walk on $\mathcal{K}$ becomes the Grover walk on the graph $\mathcal{K}$.
\end{rem}

\begin{rem}\rm
For a given $n$-dimensional admissible simplicial complex $\mathcal{K}$, we can define $d^{(k)}_i$ and ${d^{(k)}_i}^\ast$ in the same manner for the $k$-dimensional skeleton $\mathcal{K}^{(k)} = \{K_j\}_{j=0}^k$ of $\mathcal{K}$ for each $k=0,\cdots, n-1$.
\end{rem}

We go back to construction of simplicial quantum walks.
One can consider a permutation $\pi\in \mathcal{S}_{n+1}$ whose order is less than $n+1$ to define S-quantum walks. 
In this case, however, quantum walks can be trivial in the following sense.

\begin{ex}\rm
\label{ex-tether}
Let $n=2$ and $\mathcal{K}$ be the simplicial complex drawn in Figure \ref{fig-Swalk}. 
Also, let $\pi'\in \mathcal{S}_{3}$ be
\begin{equation*}
\pi' [abc] = [acb].
\end{equation*}
This permutation $\pi'$ defines an S-quantum walk $U'$ on $\mathcal{K}$ with a weight $w$. 
Consider time evolution of the state $\delta_{[abc]}^{(2)}$ under $U'$. One easily knows that the state $\delta_{[abc]}$ does not transmit states on $|abd'|$ and $|abd''|$ {\em for arbitrary weights}. In other words, states interact only between $|abc|$ and $|dbc|$ and we can observe no motions on quantum walks. 
\end{ex}

This example gives us the following problem : {\em which permutation $\pi \in \mathcal{S}_{n+1}$ associates S-quantum walk  exhibiting nontrivial behavior?} 
We propose the following concept to formulate this problems mathematically.
\begin{dfn}[Tethered quantum walks]\rm
\label{dfn-tethered}
For $\sigma = [a_0 a_1 \cdots a_n]\in \tilde K_n$, let $\{\sigma\} := \{a_0,a_1,\cdots, a_n\}$. We say that a permutation $\pi\in \mathcal{S}_{n+1}$ {\it induces tethered (S-)quantum walks} on $\mathcal{K}$ if the following statement holds: for arbitrary weights on $\mathcal{K}$ and associated quantum walks $U=S_\pi C_0$, there exists a vertex $a_j \in \{\sigma\}$ such that the following two statements are equivalent for all $m\in \mathbb{N}$:
\begin{enumerate}
\item $\tau \in \tilde K_n$ satisfies $\langle \delta^{(n)}_\tau, U^m \delta^{(n)}_\sigma \rangle_{\tilde K_n}\not = 0$.
\item $a_j\in \{\tau\}$.
\end{enumerate}
Conversely, we say that a permutation $\pi\in \mathcal{S}_{n+1}$ {\it induces movable S-quantum walks} on $\mathcal{K}$ if $\pi$ does not induce tethered quantum walks.
\end{dfn}

We easily know that the permutation $\pi' : [abc] \mapsto [acb]$ induces tethered S-quantum walks on the complex in Example \ref{ex-tether}. 
In Example \ref{ex-tether}, $a_j$ in Definition \ref{dfn-tethered} corresponds to vertices $\{b\}$ and $\{c\}$.
The walk $U_{\pi'}$ is always \lq\lq tethered" around the starting position and never get far away from there.
For these reasons, we apply $\pi \in \mathcal{S}_3$ given by (\ref{sample-permutation}) in the next section. 
By the definition of $U_\pi$, we expect that this $\pi$ induces movable quantum walks. 
As long as we observe numerical simulation results below, this permutation induces movable quantum walks. 
Our problem is then translated into determination of  the class of permutations which induce movable quantum walks.
On the other hand, a movable quantum walk does not always ensure a nontrivial behavior; 
we have to check if the quantum walk generates {\em both} transmission and reflection simultaneously. 
This is one of the other classification problems of quantum walks.

\begin{dfn}[Non-interactive quantum walks]\rm
\label{dfn-non-interactive}
A simplicial quantum walk $U$ on $\mathcal{K}$ is {\em non-interactive} if, for any $\sigma\in \tilde K_n$ and for all $n\in \mathbb{N}$, there is a unique element $\tau \in \tilde K_n$ such that $\langle \delta^{(n)}_\tau, U^m \delta^{(n)}_\sigma \rangle_{\tilde K_n}\not = 0$.
Conversely, $U$ is {\em interactive} if $U$ is not non-interactive.
\end{dfn}
In general, interactive properties of quantum walks depend on their constructions, as indicated in the No-go lemma by Meyer \cite{Meyer}. This problem is also nontrivial in the construction of quantum walks on given simplicial complexes. 
In the case of S-quantum walks, weights determine whether or not quantum walks are interactive. 
In this paper, we have two examples of the non-interactive quantum walks which exhibit only the transmission without reflection. See Figs. \ref{fig-Swalk-R2}-(a) and \ref{fig-Swalk-cylinder}-(a).
Our numerical simulations in the next section suggest that appropriate choice of weights let S-quantum walks interactive. 

\bigskip
Both notions of tethered quantum walks and non-interactivity ask us a question of appropriate construction of quantum walks. 
In other words, these notions ask us whether we can choose an appropriate weight on a given complex $\mathcal{K}$ so that quantum walk $U$ on it is movable and interactive.
Example \ref{ex-tether} indicates that mobility of walks concerns with the choice of permutations, namely, shift operators.
As indicated in the S-quantum walk on $\mathbb{R}^2$ with the weight (\ref{weight-R2-1}) in the next section, which is a non-interactive quantum walk, the interactivity of walks concerns with the choice of coin operators.


\section{Behavior of S-quantum walks -- numerical study}
\label{section-numerical}

In this section, we study how an S-quantum walk $U$ behaves on simplicial complexes. Throughout this section we only consider the case $n=2$. Our object here is to study relationships between geometry of $\mathcal{K}$ and nontrivial dynamics of $U$ such as localization. 
Here we study quantum walks on the following sample spaces:
\begin{description}
\item[\rm Subsection \ref{section-R2} : ] $\mathbb{R}^2$. The corresponding simplicial complex is $\mathcal{K}_0$. Numerical results are shown in Figures \ref{fig-Swalk-R2} and \ref{fig-Swalk-1state}.
\item[\rm Subsection \ref{section-cylinder} : ] an infinite cylinder. The corresponding simplicial complex is $\mathcal{K}_1$. Numerical results are shown in Figures \ref{fig-Swalk-cylinder}-(a), (b) and \ref{fig-Mobius}-(b).
\item[\rm Subsection \ref{section-cylinder-tetra} : ] an infinite cylinder with a tetrahedron. The corresponding simplicial complex is $\mathcal{K}_2$. Numerical results are shown in Figure  \ref{fig-Swalk-cylinder}-(c).
\item[\rm Subsection \ref{section-Mobius} : ] the M\"{o}bius band. The corresponding simplicial complex is $\mathcal{K}_3$. Numerical results are shown in Figure  \ref{fig-Mobius}-(a).

\end{description}

\subsection{S-quantum walk on $\mathbb{R}^2$}
\label{section-R2}
Here we consider one of the simplest examples, the quantum walk on $\mathbb{R}^2$.
\subsubsection{Setting}
\label{section-setting-R2}
First we represent $\mathbb{R}^2$ as a polyhedron via the simplicial decomposition. 
In the case of $\mathbb{R}^2$, one of the simplest decomposition is the uniform triangular decomposition drawn in Figure \ref{fig-R2}.
We shall use such a decomposition as a sample examination. 
Let $\mathcal{K}_0$ be the obtained simplicial complex.

\bigskip
We shall derive the algorithmic procedure of the time evolution so that readers easily follow a series of computations.
To implement the discrete-time evolution of S-quantum walk $U$, 
we take one-to-one correspondence between CONB on $\ell^2(\tilde K_2)$ and its graphical representation shown in Figures \ref{fig-R2} and \ref{fig-state}.  
See also Remark \ref{rem-notation} below. 

\begin{rem}[Notations in Case 1, 2 below and Appendix \ref{section-implementation}]\rm
\label{rem-notation}
We use the following notations for describing one time evolution of quantum walks.
They let us construct time evolution of quantum walks on concrete simplicial complexes systematically.
\begin{description}
\item[Rule 1: ] We label each simplex by integers. In the case of quantum walk on $\mathcal{K}_0$, for example, it is sufficient to consider $2N^2$ $2$-simplices for sufficiently large $N$. We then correspond each simplex to an integer from $0$ to $2N^2-1$, $\{2(iN+j)+k\mid i,j=0,\cdots, N-1, k=0,1\}$. 
In our notations, variable $i$ denotes the $x$-coordinate of indexed simplices and $j$ denotes the $y$-coordinate of indexed simplices. 
Furthermore, our complexes consist of two kinds of triangles, the lower one and the upper one (Figure \ref{fig-state}). 
$k=0,1$ correspond to lower and upper triangles, respectively. 
We label the lower triangle $m$ and the upper triangle $m+1$, 
where $m\in \{2(iN+j)\mid i,j\in \mathbb{Z}\}$ if it is defined, unless otherwise noted. 
\item[Rule 2: ] $\delta_{m,l}$ denotes the basis labeled by $m$ and $l$. 
Here $m$ is the index labeling $|\sigma|$ as an element of $K_2$, and $l\in \{0,1,\cdots, 5\}$ denotes the index determined by the rule drawn in Figure \ref{fig-state}. 
Similarly, we identify the weight $w(\sigma)$ for $\sigma \in \tilde K_2$ with $w_{m,l}$ in the same manner.
\item[Rule 3: ] In the case of quantum walks on cylinder, Subsections \ref{section-cylinder} and \ref{section-cylinder-tetra}, we impose the periodic boundary conditions on boundaries. 
In our arguments, the periodic boundary condition in the $x$-coordinate is imposed.
In the case of M\"{o}bius bands, we impose twisted identification in the $x$-coordinate. Details are shown in Subsection \ref{section-Mobius}.
\end{description}
\end{rem}

\begin{figure}[htbp]\em
\begin{minipage}{1\hsize}
\centering
\includegraphics[width=14.0cm]{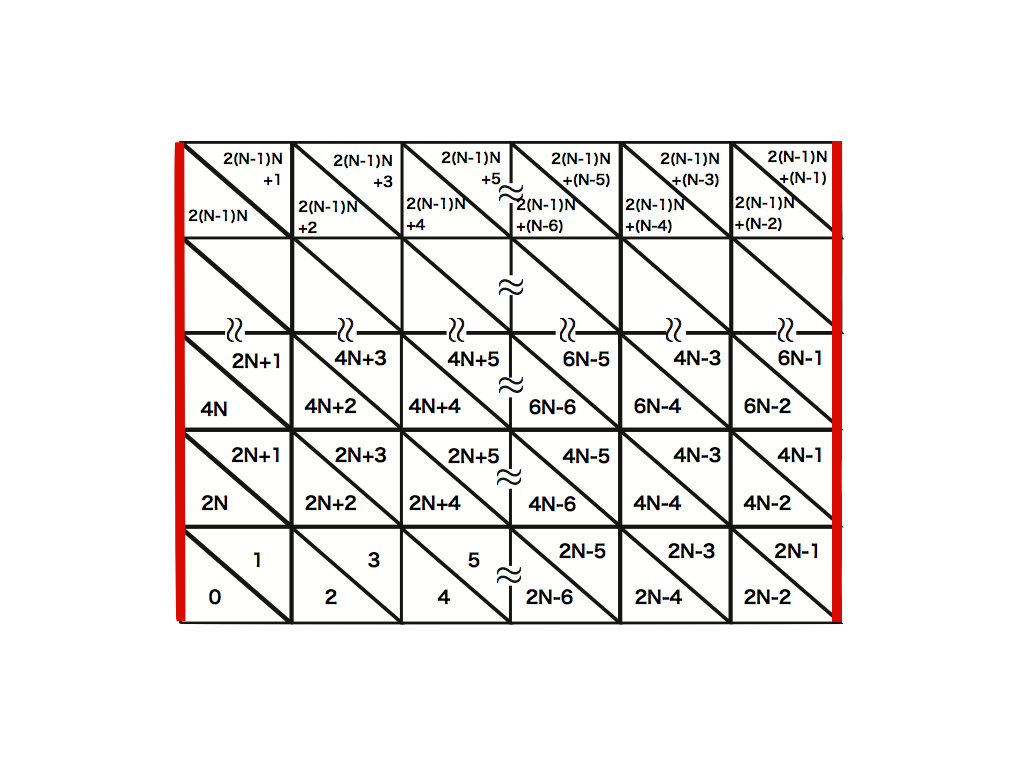}
\end{minipage}
\caption{Uniform simplicial decomposition $\mathcal{K}_0$ of $\mathbb{R}^2$.}
\label{fig-R2}
This decomposition can be realized by, say, dividing the unit square $[0,1]\times [0,1]$ into $2$ triangles. In practical computations we make such complex on $[-N,N]\times [-N,N]\subset \mathbb{R}^2$, consisting of $2N^2$ triangles. When we consider the S-quantum walk on an infinite cylinder (Subsection \ref{section-cylinder}), we impose the periodic boundary condition on, say, two lines drawn by red lines.
\end{figure}

\begin{figure}[htbp]\em
\begin{minipage}{1\hsize}
\centering
\includegraphics[width=12.0cm]{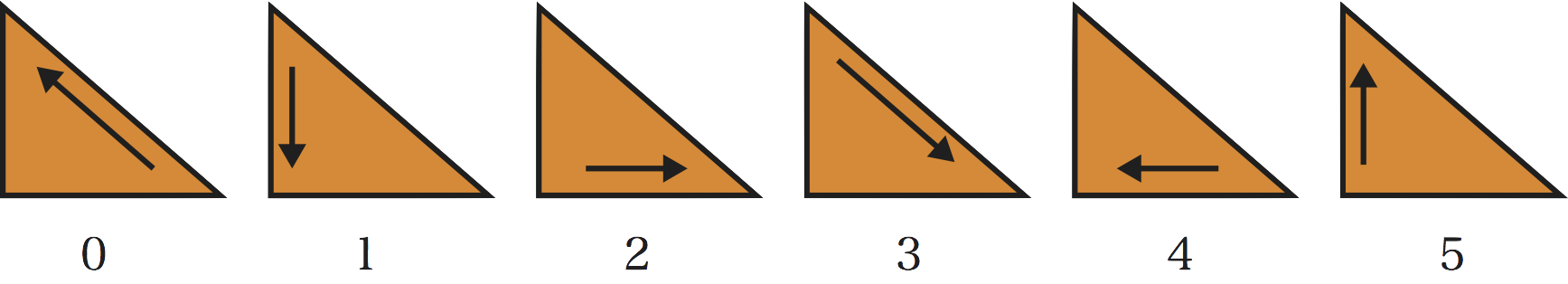}
(a)
\end{minipage}
\begin{minipage}{1\hsize}
\centering
\includegraphics[width=12.0cm]{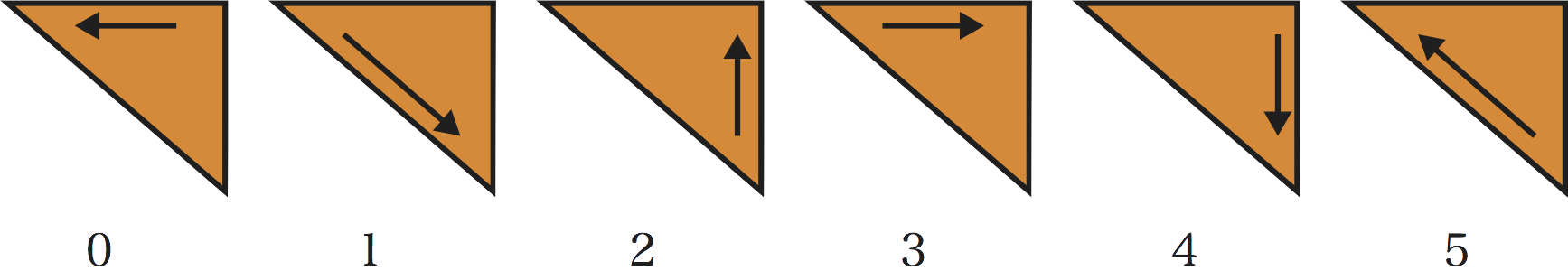}
(b)
\end{minipage}
\caption{Labeling of bases.}
\label{fig-state}
(a): The graphical representation of basis $\delta_{2(iN+j),l}$, $l\in \{0,1,2,3,4,5\}$. (b): The graphical representation of basis $\delta_{2(iN+j)+1,l}$, $l\in \{0,1,2,3,4,5\}$. 
\end{figure}

\begin{description}
\item[Case 1: the lower triangle (Figure \ref{fig-state}-(a))] 
\end{description}
By following the definition of $U$, the evolution of a state on a lower triangle is given by 
\begin{align*}
&\begin{pmatrix}
\delta_{2(iN+j),0}\\
\delta_{2(iN+j),1}\\
\delta_{2(iN+j),2}\\
\delta_{2(iN+j),3}\\
\delta_{2(iN+j),4}\\
\delta_{2(iN+j),5}
\end{pmatrix}\\
&\mapsto
\begin{pmatrix}
(2|w_{2(iN+j),0}|^2-1)\cdot \delta_{2(iN+j),1} \\
(2|w_{2(iN+j),1}|^2-1)\cdot \delta_{2(iN+j),2} \\
(2|w_{2(iN+j),2}|^2-1)\cdot \delta_{2(iN+j),0} \\
(2|w_{2(iN+j),3}|^2-1)\cdot \delta_{2(iN+j),4} \\
(2|w_{2(iN+j),4}|^2-1)\cdot \delta_{2(iN+j),5} \\
(2|w_{2(iN+j),5}|^2-1)\cdot \delta_{2(iN+j),3}
\end{pmatrix}
 +
\begin{pmatrix}
2\overline{w_{2(iN+j),0}}w_{2(iN+j)+1,5} \cdot \delta_{2(iN+j)+1,3}\\
2\overline{w_{2(iN+j),1}}w_{2((i-1)N+j)+1,4} \cdot \delta_{2((i-1)N+j)+1,5}\\
2\overline{w_{2(iN+j),2))}}w_{2(iN+(j-1))+1,3} \cdot \delta_{2(iN+(j-1))+1,4}\\
2\overline{w_{2(iN+j),3}}w_{2(iN+j)+1,1} \cdot \delta_{2(iN+j)+1,2}\\
2\overline{w_{2(iN+j),4}}w_{2(iN+(j-1))+1,0} \cdot \delta_{2(iN+(j-1))+1,1}\\
2\overline{w_{2(iN+j),5}}w_{2((i-1)N+j)+1,2} \cdot \delta_{2((i-1)N+j)+1,0}
\end{pmatrix}.
\end{align*}
The first term of the right hand side represents the \lq\lq reflection" of waves and the second represents the \lq\lq transmission" of waves to the adjacent upper triangle.
\begin{description}
\item[Case 2: the upper triangle (Figure \ref{fig-state}-(b))] 
\end{description}
Similarly, the evolution of the state on an upper triangle is given by 
\begin{align*}
&\begin{pmatrix}
\delta_{2(iN+j)+1,0}\\
\delta_{2(iN+j)+1,1}\\
\delta_{2(iN+j)+1,2}\\
\delta_{2(iN+j)+1,3}\\
\delta_{2(iN+j)+1,4}\\
\delta_{2(iN+j)+1,5}
\end{pmatrix}\\
&\mapsto
\begin{pmatrix}
(2|w_{2(iN+j)+1,0}|^2-1)\cdot \delta_{2(iN+j)+1,1} \\
(2|w_{2(iN+j)+1,1}|^2-1)\cdot \delta_{2(iN+j)+1,2} \\
(2|w_{2(iN+j)+1,2}|^2-1)\cdot \delta_{2(iN+j)+1,0} \\
(2|w_{2(iN+j)+1,3}|^2-1)\cdot \delta_{2(iN+j)+1,4} \\
(2|w_{2(iN+j)+1,4}|^2-1)\cdot \delta_{2(iN+j)+1,5} \\
(2|w_{2(iN+j)+1,5}|^2-1)\cdot \delta_{2(iN+j)+1,3}
\end{pmatrix}
 +
\begin{pmatrix}
2\overline{w_{2(iN+j)+1,0}}w_{2(iN+(j+1)),4} \cdot \delta_{2(iN+(j+1)),5}\\
2\overline{w_{2(iN+j)+1,1}}w_{2(iN+j),3} \cdot \delta_{2(iN+j),4}\\
2\overline{w_{2(iN+j)+1,2}}w_{2((i+1)N+j),5} \cdot \delta_{2((i+1)N+j),3}\\
2\overline{w_{2(iN+j)+1,3)}}w_{2(iN+(j+1)),2} \cdot \delta_{2(iN+(j+1)),0}\\
2\overline{w_{2(iN+j)+1,4}}w_{2((i+1)N+j),1} \cdot \delta_{2((i+1)N+j),2}\\
2\overline{w_{2(iN+j)+1,5}}w_{2(iN+j),0} \cdot \delta_{2(iN+j),1}
\end{pmatrix}.
\end{align*}
As in the case of lower triangles, the first term of the right hand side represents the \lq\lq reflection" of waves and the second represents the \lq\lq transmission" of waves to the adjacent lower triangle.

\bigskip
All states evolve according to one of the above rules. We are then ready to compute S-quantum walk on $\mathbb{R}^2$.

\subsubsection{S-quantum walk on $\mathbb{R}^2$}
Let the triangle $|abc|\in K_2$ be located at the center of $N\times N$ lattice, in particular, $\delta^{(2)}_{\pi^i [abc]} = \delta_{\frac{N}{2}\times N + \frac{N}{2}+0, i}$.
We set the initial state $\Psi_0 = \sum_{i=0}^2\varphi_{\pi^i [abc]} \delta^{(2)}_{\pi^i [abc]} + \sum_{i=0}^2\varphi_{\pi^i [acb]} \delta^{(2)}_{\pi^i [acb]}\in \ell^2(\tilde K_2)$, where 
\begin{equation*}
\varphi_{\pi^i [abc]} = \varphi_{\pi^i [acb]} = 1/\sqrt{6},\quad i=0,1,2.
\end{equation*}
The S-quantum walk with the identical weight, namely, 
\begin{equation}
\label{weight-R2-1}
w(\sigma)\equiv 1/\sqrt{2},\quad \forall \sigma\in \tilde K_2,
\end{equation} 
is shown in Figure \ref{fig-Swalk-R2}-(a). 

Next, we change the weight into
\begin{align}
\notag
w(\sigma_1) &\equiv \sqrt{1/3}\quad \text{ for all $\sigma_1\in \tilde K_2$ with $|\sigma_1|$ being a lower triangle},\\
\label{weight-R2-2}
w(\sigma_2) &\equiv \sqrt{2/3}\quad \text{ for all $\sigma_2\in \tilde K_2$ with $|\sigma_2|$ being an upper triangle}.
\end{align}
The computation result is drawn in \ref{fig-Swalk-R2}-(b). Additionally, the quantum walk with the weight
\begin{align}
\notag
w(\sigma_1) &\equiv \sqrt{1/10}\quad \text{ for all $\sigma_1\in \tilde K_2$ with $|\sigma_1|$ being a lower triangle},\\
\label{weight-R2-3}
w(\sigma_2) &\equiv \sqrt{9/10}\quad \text{ for all $\sigma_2\in \tilde K_2$ with $|\sigma_2|$ being an upper triangle}
\end{align}
is drawn in Figure \ref{fig-Swalk-R2}-(c). In any cases, the state spreads in six-directions drawing {\em a hexagram}. 
In Figure \ref{fig-Swalk-R2}-(a), no reflection occurs during time evolutions. 
These pictures imply that the {\em ballistic} spreading, which is the strongest spreading, is exhibited in S-quantum walks with uniform weights. 
Note that this is a similar behavior to the two-state Grover walk on $\mathbb{Z}$. 
This walk with the weight (\ref{weight-R2-1}) is thus a non-interactive quantum walk. However, we can adjust the weight such as (\ref{weight-R2-2}) and (\ref{weight-R2-3}) to make the interactions. See Figs. \ref{fig-Swalk-R2}-(b) and (c).
These pictures also imply that, the bigger the difference between $w(\sigma_1)$ and $w(\sigma_2)$ is, the smaller the size of hexagram is. This observation is related to the pseudo velocity of the linear spreading.

Next we fix the weight $w$ by (\ref{weight-R2-2}) and change the initial state into $\Psi_0 = \delta^{(2)}_{[abc]}\in \ell^2(\tilde K_2)$. 
Resulting behavior is shown in Figure \ref{fig-Swalk-1state}. In this case, the state spreads drawing a triangle. 
Similarly, the quantum walk with the initial state $\Psi_0 = \frac{1}{\sqrt{2}}\delta^{(2)}_{[abc]} + \frac{1}{\sqrt{2}}\delta^{(2)}_{[bca]} = \frac{1}{\sqrt{2}}\delta^{(2)}_{[abc]} + \frac{1}{\sqrt{2}}\delta^{(2)}_{\pi [abc]}\in \ell^2(\tilde K_2)$ spreads linearly. 
Its support is also a similar triangle as shown in Figure \ref{fig-Swalk-1state}-(b). 
On the other hand, if we consider the quantum walk with the initial state $\Psi_0 = \frac{1}{\sqrt{2}}\delta^{(2)}_{[abc]} + \frac{1}{\sqrt{2}}\delta^{(2)}_{[acb]}\in \ell^2(\tilde K_2)$, then the walk spreads drawing a hexagram like Figure \ref{fig-Swalk-1state}-(c). 

\bigskip

\subsubsection{Comparison with quantum walk on the triangular lattice}
We compare the behavior of the Grover walk on the triangular lattice with S-quantum walk on $\mathbb{R}^2$. 
The asymptotic behavior of the quantum walk on lattices, in particular, the Grover walk is studied not only numerically \cite{TFMK} but mathematically \cite{HKSS2}. 
It is worth studying how the difference of dimension of complexes affect dynamics. 
Our criteria are focused on important and essential properties of quantum walks, {\em localization} and {\em linear spreading}.
\begin{rem}\rm
Here we review localization and linear spreading of quantum walks on a lattice $L$.
\begin{enumerate}
\item We say that {\em localization} occurs if there exists a point ${\bf x}\in L$ such that
\begin{equation*}
\limsup_{n\to \infty} \mu_n({\bf x}) > 0,
\end{equation*}
where $\mu_n$ is the probability of the quantum walk $\Psi_n$ at time $n$ and position ${\bf x}\in L$ given by
\begin{equation*}
\mu_n({\bf x}) := \sum_{e\in \tilde K_1:o(e)={\bf x}}\|\Psi_n(e)\|^2_{\tilde K_1}.
\end{equation*}
Here $o(e)$ is the origin vertex of the directed edge $e$. For example, if $e=ab$, then $o(e)=a$. Similarly, if $e=ba$, then $o(e)=b$. See \cite{HKSS2} for details.
\item We say that {\em linear spreading} happens if
\begin{equation}
\label{linear-lattice}
\lim_{n\to \infty} \frac{V_n}{n^2} \in (0,\infty),
\end{equation}
where $V_n$ is the radial variance of the quantum walk $\Psi_n$ defined by
\begin{equation}
\label{variance-lattice}
V_n:= \left\{ \sum_{x\in L}\|{\bf x}\|^2 \mu_n({\bf x})^2 - \left(\sum_{x\in L}\|{\bf x}\| \mu_n({\bf x})\right)^2\right\}
\end{equation}
and $\|{\bf x}\|$ is the Euclidean norm of ${\bf x}$.
\end{enumerate}
The relationship (\ref{linear-lattice}) corresponds to weak convergence of the random variable $\Psi_n/n$. The denominator $n$ corresponds to the linearity of spreading rates, which is a quite different property from classical random walks.
\end{rem}

In \cite{HKSS2}, it is shown that the Grover walk on the triangular lattice exhibits {\em both} localization and linear spreading, as shown in Figure \ref{fig-QW-triangle}. 
On the contrast, we cannot observe localization in the case of S-quantum walks on $\mathcal{K}_0$, which can be seen in Figure \ref{fig-Swalk-R2} as noted before. This is a crucial difference between quantum walk on the triangular lattice and on $\mathcal{K}_0$. Such an observation can be stated more topologically. 
In \cite{HKSS2}, it is shown that the Grover walk on a finite graph $G$ exhibits localization if $G$ possesses {\em a cycle} in the homological sense. In other words, nontrivial first homology classes of $G$ induce localization. 
On the other hand, the simplicial complex $\mathcal{K}_0$ has no \lq\lq cycles", that is, topological holes such as rings and cavities. 
One can thus expect that the dimension and topology of simplicial complexes have great influences on the behavior of quantum walks.

\bigskip
Next, in order to study the spreading rate of $U$ on $\mathcal{K}_0$, we compute the radial variance
\begin{equation}
\label{variance-complex}
V_n = V_n^{(f)}:= \left\{\sum_{|\sigma|\in K_2}\sum_{{\bf x} =(x,y)\in |\sigma|} \|{\bf x}\|^2 \mu_n^{(f)}(|\sigma|)^2 - \left(\sum_{|\sigma|\in  K_2}\sum_{{\bf x} =(x,y)\in |\sigma|}\|{\bf x}\| \mu_n^{(f)}(|\sigma|) \right)^2\right\},
\end{equation}
corresponding to (\ref{variance-lattice}). Here $f$ denotes an initial state, but we often drop the letter $(f)$ if initial states are clear in the context.
In practical computations, we condense the density and the probability $\mu_n$ of each simplex on a point, say, a vertex for simplicity. 
The time evolution of $V_n/n^2$ is drawn in Figure \ref{fig-spread}-(a). This figure shows that the quantity $V_n/n^2$ is close asymptotically to a certain positive value like (\ref{linear-lattice}), which implies that {\em linear spreading happens for S-quantum walk on $\mathcal{K}_0$}.

\subsection{S-quantum walk on an infinite cylinder}
\label{section-cylinder}
Next we consider the S-quantum walk on an infinite cylinder $S^1\times \mathbb{R}$. 
Note that an infinite cylinder possesses a topological cycle, more precisely, the nontrivial first homology class (see Figure \ref{fig-homology} in Appendix \ref{section-complexes} for quick and intuitive understanding of homology). 
In the case of Szegedy walks on graphs, nontrivial cycle structure can induce localization \cite{HKSS2}. 
Our focus in this section is to study if the nontrivial homology of simplicial complexes can induce localization.
As in the case of $\mathbb{R}^2$, we derive a triangular decomposition of $S^1\times \mathbb{R}$ as the first step.

\subsubsection{Setting}
First we prepare a large rectangular domain with the uniform triangular decomposition. In this case we set $R\times N$-pairs of lower and upper triangles, where $R$ is a positive integer not so large, say, $10$ and $N$ is a large integer, say, $2000$. We additionally impose the periodic boundary condition on the $x$-axis. See Figure \ref{fig-R2}. The time evolution rule is the same as in the case of $\mathbb{R}^2$ with the additional periodic boundary condition. Let $\mathcal{K}_1$ be the obtained simplicial complex.

\subsubsection{S-quantum walk on an infinite cylinder}
We set the initial state $\Psi_0 = \sum_{i=0}^2\varphi_{\pi^i [abc]} \delta^{(2)}_{\pi^i [abc]} + \sum_{i=0}^2\varphi_{\pi^i [acb]} \delta^{(2)}_{\pi^i [acb]}\in \ell^2(\tilde K_2)$, where 
\begin{equation}
\label{initial-cylinder}
\varphi_{\pi^i [abc]} = \varphi_{\pi^i [acb]} = 1/\sqrt{6},\quad i=0,1,2.
\end{equation}
The S-quantum walk with the identical weight, namely, $w(\sigma)\equiv 1/\sqrt{2}$ is shown in Figure \ref{fig-Swalk-cylinder}-(a). 
Next, we change the weight into (\ref{weight-R2-2}), namely,
\begin{align*}
w(\sigma_1) &\equiv \sqrt{1/3}\quad \text{ for all $\sigma_1\in \tilde K_2$ with $|\sigma_1|$ being a lower triangle},\\
w(\sigma_2) &\equiv \sqrt{2/3}\quad \text{ for all $\sigma_2\in \tilde K_2$ with $|\sigma_2|$ being an upper triangle}.
\end{align*}
The computation result is depicted by Figure \ref{fig-Swalk-cylinder}-(b). 
In both cases, we can observe a positive probability density on the central circle on which the initial state is located. 
These pictures indicate that {\em localization occurs for $U$ on $S^1\times \mathbb{R}$}. 
Now we compute {\em the time-averaged probability} of S-quantum walks with the initial state $f$ defined by
\begin{equation}
\label{ave-meas}
\bar \mu_T(|\sigma|) \equiv \bar \mu_T^{(f)}(|\sigma|) := \frac{1}{T}\sum_{m=0}^{T-1} \mu_m^{(f)}(|\sigma|) ,\quad |\sigma| \in K_2.
\end{equation}
The calculation result is shown in Figure \ref{fig-time-average}. The time-averaged probability $\bar \mu_T(|abc|)$ converges to a positive value as $T\to \infty$, which leads to the other suggestion that the localization occurs at the starting simplex $|abc|$. 

As S-quantum walks on $\mathbb{R}^2$, time evolutions of the radial variance $V_n/n^2$ with different weights are shown in Figure \ref{fig-spread}-(c) and (d), where $V_n$ is given by (\ref{variance-complex}). 
These figures show that the quantity $V_n/n^2$ converges to a certain positive value for quantum walks on $\mathcal{K}_1$. These observations imply that {\em linear spreading happens on the S-quantum walk on $S^1\times \mathbb{R}$}. 

\begin{rem}\rm
In the case of $S^1\times \mathbb{R}$, there is a direction of the quantum walk such that a state stays in a compact subset of $S^1\times \mathbb{R}$, which reflects a bounded ring structure of the cylinder, that is, a generator of the nontrivial first homology group
$H_1(S^1\times \mathbb{R};\mathbb{Z})\cong \mathbb{Z}$.
\end{rem}

\subsection{S-quantum walk on an infinite cylinder with a tetrahedron}
\label{section-cylinder-tetra}
In the third example, we consider the S-quantum walk on $\mathcal{K}_2$: an infinite cylinder discussed in the previous subsection with a tetrahedron. Tetrahedra have nontrivial cavities which generate the second homology classes. 
Here we consider the effect of such homology classes on S-quantum walks.

\subsubsection{Setting}
As in previous subsections, we construct a simplicial complex as the first step. We prepare the simplicial complex $\mathcal{K}_1$  and attach a tetrahedron somewhere on $\mathcal{K}_1$ to obtain the simplicial complex $\mathcal{K}_2$. 
The brief illustration of $\mathcal{K}_2$ around a tetrahedron is shown in Figure \ref{fig-tetrahedron}. 
The labeling of $2$-simplices are derived as in the case of $\mathcal{K}_1$ as well as the description in Figure \ref{fig-tetrahedron}. One can realize this construction in the standard manner and we omit the detail.

The additional algorithmic procedure of S-quantum walks on $\mathcal{K}_2$ is listed in Appendix \ref{section-implementation}.

\subsubsection{S-quantum walk on an infinite cylinder with a tetrahedron}
We set the initial state $\Psi_0\in \ell^2(\tilde K_2)$ by (\ref{initial-cylinder}).
Note that the $2$-simplex $|abc|$ is now set by a part of the tetrahedron in $\mathcal{K}_2$. See Figure \ref{fig-tetrahedron}.
We then fix the weight 
\begin{align}
\label{weight-Grover}
&w(\sigma)\equiv 1/\sqrt{\deg (\tilde d_0^{(2)} \sigma)},
\quad \text{ where }\deg (\tilde d_0^{(2)} \sigma):= \sharp\{\tau \in \tilde K_1: \tilde d_0^{(2)} \sigma = \tau, \sigma \in \tilde K_2\}.
\end{align}
In the case of quantum walks on graphs, such weight corresponds to the Grover walk on graphs. 

The S-quantum walk with the weight (\ref{weight-Grover}) and the initial state $\Psi_0$ is shown in Figure \ref{fig-Swalk-cylinder}-(c). 
In this example, localization cannot be observed around the central cycle comparing with the previous example on $\mathcal{K}_1$. 
Instead, we can see the localization on the \lq\lq bump", that is, the tetrahedron.
Indeed, in the early stage the state with positive probability density on the circle is observable. 
After sufficient large time evolutions, such a state is absorbed by the tetrahedron. 
Such a phenomenon can be observed from time-averaged probability $\bar \mu_T$ given in (\ref{ave-meas}). 
The measure $\bar \mu_T$ at the starting simplex $|abc|$ converges to a positive value as $T\to \infty$ as shown in Figure \ref{fig-time-average}. 
On the other hand, the measure at a simplex off the tetrahedron but on the central cycle tends to zero, while the measure for S-quantum walk on $\mathcal{K}_1$ at the simplex on the central cycle tends to a positive value. 
Compare the graph \lq\lq tetra-off" with \lq\lq cylinder-transmit-adjacent" in Figure \ref{fig-time-average}.
These observations imply that the higher dimensional homology classes absorb states running on lower dimensional homology classes. 
In other words, the higher dimensional homology classes induce the stronger localization.


\subsection{S-quantum walk on the M\"{o}bius band}
\label{section-Mobius}
In the last example, we consider S-quantum walks on the M\"{o}bius band. 
As in the case of $\mathbb{R}^2$, we derive a triangular decomposition $\mathcal{K}_3$ of the M\"{o}bius band as the first step. 
Construction of $\mathcal{K}_3$ is the same as $\mathcal{K}_1$ except the identification of the boundary. 
In general, the infinite M\"{o}bius band is constructed identifying $(0,t)\in [0,1]\times \mathbb{R}$ with $(1,-t)\in [0,1]\times \mathbb{R}$ in the infinite strip $[0,1]\times \mathbb{R}$. 
With this in mind, we construct $\mathcal{K}_3$ and S-quantum walks on it. 
It is well known that the M\"{o}bius band is non-orientable in the sense of differentiable manifolds or vector bundles, while it has the same homology as an infinite cylinder. 
See e.g. \cite{MS}. Our interest here is whether the orientability of simplicial complexes affect dynamics of S-quantum walks.

\subsubsection{Setting}
First we prepare a large rectangular domain with the uniform triangular decomposition. In this case we set $R\times N$-pairs of lower and upper triangles, where $R$ is a positive integer not so large, say, $8$ and $N$ is a large integer, say, $2500$. We additionally impose the twisted boundary condition on the $x$-axis. See Figure \ref{fig-Mobius-rule}. The time evolution rule is the same as in the case of $\mathbb{R}^2$ with additional twisted boundary condition. Note that one-periodic motion of quantum walks in $x$-direction changes the chirality, in other words, the rotating direction of arrows. This is because the M\"{o}bius band is non-orientable. 

\begin{figure}[htbp]\em
\begin{minipage}{0.5\hsize}
\centering
\includegraphics[width=8.0cm]{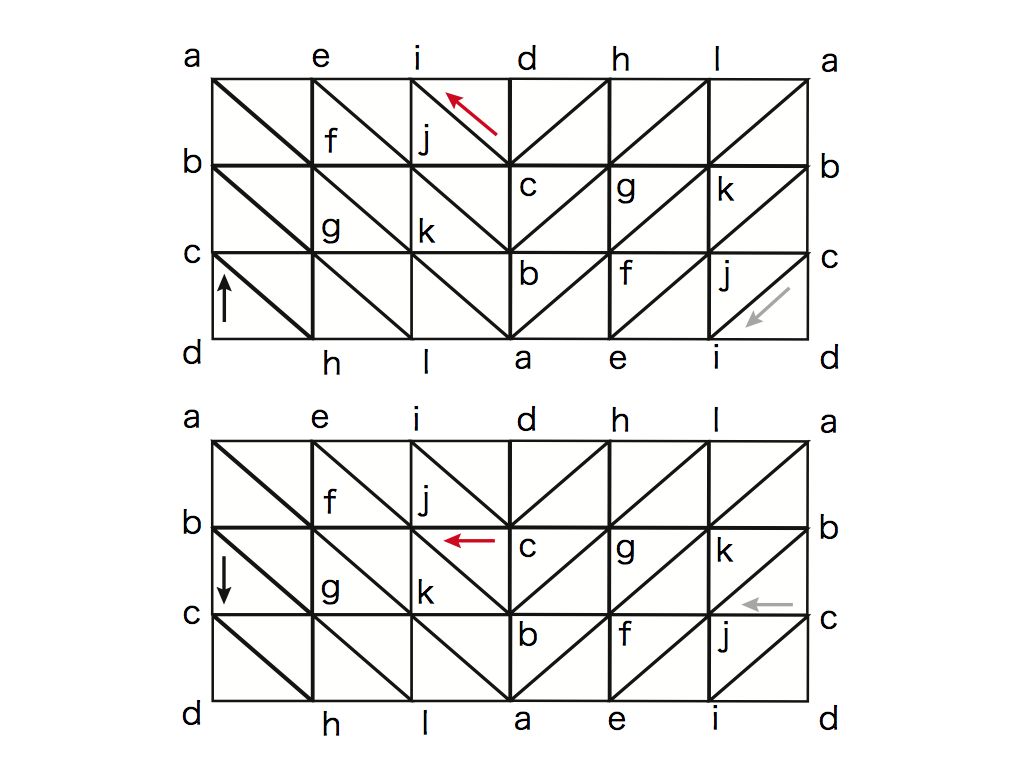}
\end{minipage}
\begin{minipage}{0.5\hsize}
\centering
\includegraphics[width=8.0cm]{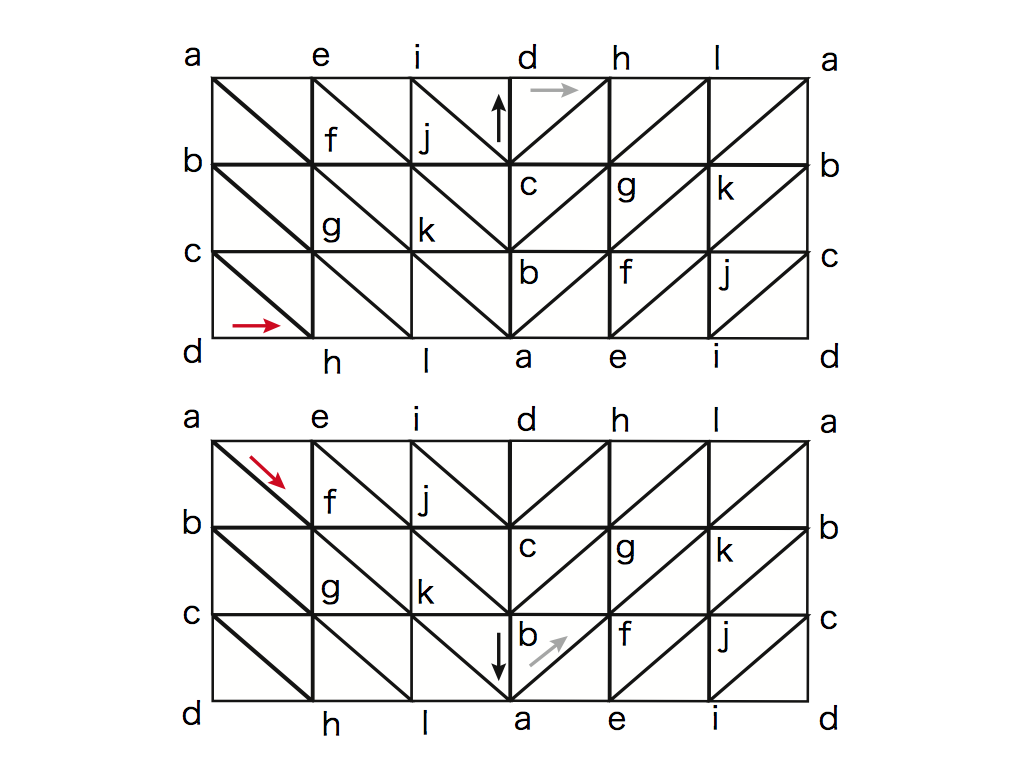}
\end{minipage}
\caption{Implementation of quantum walks on the M\"{o}bius band.}
\label{fig-Mobius-rule}
These figures show implementations of the construction of $\mathcal{K}_3$ and S-quantum walks on it.
In these figures, the twisted identification of the boundary is imposed on $|ab|, |bc|$ and $|cd|$. This identification changes the rotating direction of triangles. See $|abf|$, for example. $\mathcal{K}_3$ can be constructed by this manner.

These figures describe how quantum walks behave on $\mathcal{K}_3$. 
Let black arrows be incident states. 
The definition of S-quantum walks yields transmitted states drawn by grey arrows, if necessary, via the periodic identification on the boundary.
Grey arrows actually describe state on the copy of complexes. 
True transmitted states are red arrows via identification of the copy. 
\end{figure}

\subsubsection{S-quantum walk on the M\"{o}bius band}
We set the initial state $\Psi_0 = \sum_{i=0}^2\varphi_{\pi^i[abc]} \delta^{(2)}_{\pi^i[abc]} + \sum_{i=0}^2\varphi_{\pi^i[acb]} \delta^{(2)}_{\pi^i[acb]}\in \ell^2(\tilde K_2)$ at the center of $\mathcal{K}_3$, where $\varphi_{\pi^i[abc]} = \varphi_{\pi^i[acb]} = 1/\sqrt{6},\ i=0,1,2$.
Let the weight $w(\sigma)$ be
\begin{align}
\notag
w(\sigma_1) &\equiv \sqrt{0.9/2}\quad \text{ for all $\sigma_1\in \tilde K_2$ with $|\sigma_1|$ being a lower triangle},\\
\label{weight-Mobius}
w(\sigma_2) &\equiv \sqrt{1.1/2}\quad \text{ for all $\sigma_2\in \tilde K_2$ with $|\sigma_2|$ being an upper triangle}.
\end{align}
The S-quantum walk with this weight is shown in Figure \ref{fig-Mobius}-(a). 
We observe that quantum walker goes away from the center drawing spirals. This behavior is completely different from the case of infinite cylinders, although the M\"{o}bius band is homologically identical with an infinite cylinder. 
This observation implies that quantum walks on the M\"{o}bius band does {\em not} admit localization.
This suggestion can be also observed from the viewpoint of the time-averaged probability. 
Consider the state at the starting simplex $|\sigma_0| = |abc|$, which is labelled by $(i,j,k) = (N/2, N/2, 0)$, and the adjacent simplex $|\sigma_0'|$ labelled by $(i,j,k) = (N/2, N/2, 1)$. 
Note that the support of the initial state $\Psi_0$ is $|\sigma_0|$ and that $|\sigma_0|$ is on the central cycle where localization is observed in the case of $\mathcal{K}_1$. Figure \ref{fig-time-average} shows that both $\bar \mu_T^{(\Psi_0)}(|\sigma_0|)$ and $\bar \mu_T^{(\Psi_0)}(|\sigma_0'|)$ tends to zero as $T\to \infty$.

As a comparison, quantum walks on $\mathcal{K}_1$ is shown in Figure \ref{fig-Mobius}-(b). In this figure we impose the same initial condition and weights as the case of Figure \ref{fig-Mobius}-(a). 
As we have seen in Subsection \ref{section-cylinder}, quantum walks on $\mathcal{K}_1$ has localization, which can be seen from the positive probability density at the center in $y$-direction. 

Note again that the M\"{o}bius band has the identical homology with infinite cylinder. Observations in this subsection imply that more detailed geometric features, like orientation, than homology can affect dynamics of our quantum walks.

\section{Conclusion and Discussion}
\label{section-conclusion}
In this paper we constructed a new type of quantum walks on simplicial complexes. 
We also discussed behavior of our quantum walks, simplicial quantum walks, with numerical simulations. 
We numerically observe the following.
\begin{itemize}
\item The permutation $\pi : [abc]\mapsto [bca]$ induces movable S-quantum walks (Definition \ref{dfn-tethered}).
\item Appropriate choice of weights induce interactive S-quantum walks (Definition \ref{dfn-non-interactive}). 
\item Linear spreading of S-quantum walks happen as in the case of quantum walks on lattices such as $\mathbb{Z}^d$.
\item S-quantum walks on $\mathcal{K}_0$, the triangulation of $\mathbb{R}^2$, do not admit localizations. 
This is completely different from quantum walks on the triangular lattice embedded in $\mathbb{R}^2$. 
In other words, traditional quantum walks on $\mathbb{Z}^d$ and simplicial quantum walks exhibit different behavior as traditional multi-dimensional quantum walks.
\item Localization occurs when the simplicial complex has non-trivial homological structures, namely, rings or cavities.
\item Localizations caused by the second homology classes (i.e. cavities) absorb those by the first homology classes (i.e. rings).
\item Localizations depend on the orientability of polyhedra. 
\end{itemize}
Our simplicial quantum walks have the same main features, like linear spreading and localizations, as walks on graphs. 
In such a sense, our simplicial quantum walks are ones of multi-dimensional analogues of quantum walks on graphs.
We also observe different phenomena from quantum walks on graphs, which reflects geometry of simplicial complexes. 
All observations herein can also hold for quantum walks on multi-dimensional simplicial complexes.

On the other hand, our simplicial quantum walks give us a nontrivial problem concerning exhibition of nontrivial behavior, namely, whether $\pi\in \mathcal{S}_{n+1}$ induces tethered or movable quantum walks. 
As mentioned in the end of Section \ref{section-construction},
an $n$-dimensional quantum walk can be trivial in the sense that $\pi\in \mathcal{S}_{n+1}$ induces tethered quantum walks, if the walk is associated with $\pi\in \mathcal{S}_{n+1}$ whose order is less than $n+1$.
Such a phenomenon cannot be seen in typical quantum walks on graphs, such as $2$-state or $3$-state Grover walk on $\mathbb{Z}$.
As mentioned above, we numerically observed that $\pi : [abc]\mapsto [bca]$ induces movable quantum walks. 
Even for permutations of the form (\ref{sample-permutation}), however, it is in general nontrivial if a given permutation $\pi$ induces movable quantum walks. 
This classification problem is very important to construct simplicial quantum walks since, 
as our observations with $n=2$ show, it is intrinsically related to non-triviality of quantum walks. 
In particular, tethered quantum walks concerns with non-triviality of walks from the viewpoint of shift operators, while non-interactivity concerns with non-triviality from the viewpoint of coin operators.


\bigskip
We end this paper proposing further directions of our arguments.

\begin{description}
\item[(1) Characterization of the spectrum of quantum walk $U$ and localization.] 
\end{description}
The spectrum of $U$ is a core for understanding quantum walks, in particular, related to localization. 
In \cite{HKSS2} and preceding works therein, cycles of graphs can characterize spectrum of $U$ which corresponds to localization. Our numerical results imply that such correspondence are valid for simplicial quantum walks. 
Our numerical results also show that deeper geometric features of simplicial complexes than graphs, like homology of the higher order and orientations, can affect behavior of quantum walks.
One can guess that spectrum of $U$ reveals geometry of underlying spaces as well as multi-dimensional aspect of quantum walks, which will be seen in the forthcoming paper \cite{MOS}.

\begin{description}
\item[(2) Asymptotic behavior of S-quantum walks.]
\end{description}
As quantum walks on $\mathbb{Z}^d$ or graphs, the asymptotic behavior of quantum walks is also the heart for understanding S-quantum walks. 
The weak convergence of distributions is the key to describe exhibition of linear spreading of quantum walks (e.g. \cite{K1}). Fortunately, we numerically observe that the linear spreading happens in many examples. This fact implies that various techniques for the asymptotic behavior of quantum walks on graphs can be also applied to our quantum walks, which will be seen in the forthcoming paper \cite{MOS}.

\section*{Acknowledgements}
KM was partially supported by Coop with Math Program, a commissioned project by MEXT. OO was partially supported by JSPS KAKENHI Grant Number 24540208. ES was partially supported by JSPS Grant-in-Aid for Young Scientists (B) (No. 25800088).

\bibliographystyle{jplain}
\bibliography{QW-simplicial}

\begin{thebibliography}{10}

\bibitem{Gu}
S.P.~Gudder.
\newblock {\em Quantum probability}.
\newblock Probability and Mathematical Statistics. Academic Press, Inc.,
  Boston, MA, 1988.

\bibitem{FH}
R.P.~Feynman and A.R.~Hibbs.
\newblock {\em Quantum mechanics and path integrals}.
\newblock Dover Publications, Inc., Mineola, NY, emended edition, 2010.
\newblock Emended and with a preface by Daniel F. Styer.

\bibitem{K1}
N.~Konno.
\newblock Quantum random walks in one dimension.
\newblock {\em Quantum Inf. Process.}, Vol.~1, No.~5, pp. 345--354 (2003),
  2002.

\bibitem{A1}
A.~Ambainis.
\newblock Quantum walks and their algorithmic applications.
\newblock {\em International Journal of Quantum Information}, Vol.~1, No.~04,
  pp. 507--518, 2003.

\bibitem{K2}
N.~Konno.
\newblock Quantum walks.
\newblock In {\em Quantum potential theory}, Vol. 1954 of {\em Lecture Notes in
  Math.}, pp. 309--452. Springer, Berlin, 2008.

\bibitem{A2}
A.~Ambainis.
\newblock Quantum walk algorithm for element distinctness.
\newblock {\em SIAM J. Comput.}, Vol.~37, No.~1, pp. 210--239 (electronic),
  2007.

\bibitem{AKR}
A.~Ambainis, J.~Kempe and A.~Rivosh.
\newblock Coins make quantum walks faster.
\newblock In {\em Proceedings of the {S}ixteenth {A}nnual {ACM}-{SIAM}
  {S}ymposium on {D}iscrete {A}lgorithms}, pp. 1099--1108 (electronic). ACM,
  New York, 2005.

\bibitem{SKW}
N.~Shenvi, J.~Kempe and K.B.~Whaley.
\newblock Quantum random-walk search algorithm.
\newblock {\em Physical Review A}, Vol.~67, No.~5, p. 052307, 2003.

\bibitem{CBS}
C.M.~Chandrashekar, S.~Banerjee and R.~Srikanth.
\newblock Relationship between quantum walks and relativistic quantum
  mechanics.
\newblock {\em Physical Review A}, Vol.~81, No.~6, p. 062340, 2010.

\bibitem{S}
F.W.~Strauch.
\newblock Relativistic effects and rigorous limits for discrete- and
  continuous-time quantum walks.
\newblock {\em J. Math. Phys.}, Vol.~48, No.~8, pp. 082102, 27, 2007.

\bibitem{KFCSAMW}
M.~Karski, L.F{\"o}rster, J.-M.~Choi, A.~Steffen, W.~Alt,
  D.~Meschede and A.~Widera.
\newblock Quantum walk in position space with single optically trapped atoms.
\newblock {\em Science}, Vol. 325, No. 5937, pp. 174--177, 2009.

\bibitem{MY}
L.~Matsuoka and K.~Yokoyama.
\newblock Physical implementation of quantum cellular automaton in a diatomic
  molecule.
\newblock {\em Journal of Computational and Theoretical Nanoscience}, Vol.~10,
  No.~7, pp. 1617--1620, 2013.

\bibitem{ZKGSBR}
F.~Z{\"a}hringer, G.~Kirchmair, R.~Gerritsma, E.~Solano, R.~Blatt and C.F.~Roos.
\newblock Realization of a quantum walk with one and two trapped ions.
\newblock {\em Physical review letters}, Vol. 104, No.~10, p. 100503, 2010.

\bibitem{MW}
J.~Wang and K.~Manouchehri.
\newblock {\em Physical implementation of quantum walks}.
\newblock Springer, 2013.

\bibitem{BBW}
S.~D. Berry, P.~Bourke and J.B.~Wang.
\newblock {\em qwviz}: Visualisation of quantum walks on graphs.
\newblock {\em Computer Physics Communications}, Vol. 182, No.~10, pp.
  2295--2302, 2011.

\bibitem{W}
J.~Watrous.
\newblock Quantum simulations of classical random walks and undirected graph
  connectivity.
\newblock In {\em Computational Complexity, 1999. Proceedings. Fourteenth
  Annual IEEE Conference on}, pp. 180--187. IEEE, 1999.

\bibitem{HKSS1}
Yu.~Higuchi, N.~Konno, I.~Sato and E.~Segawa.
\newblock Quantum graph walks {I}: {M}apping to quantum walks.
\newblock {\em Yokohama Math. J.}, Vol.~59, pp. 33--55, 2013.

\bibitem{Sze}
M.~Szegedy.
\newblock Quantum speed-up of markov chain based algorithms.
\newblock In {\em Foundations of Computer Science, 2004. Proceedings. 45th
  Annual IEEE Symposium on}, pp. 32--41. IEEE, 2004.

\bibitem{PM}
G.D.~Paparo and M.A.~Martin-Delgado.
\newblock Google in a quantum network.
\newblock {\em Scientific reports}, Vol.~2, , 2012.

\bibitem{PMCM}
G.D.~Paparo, M.~M{\"u}ller, F.~Comellas and M.A.~Martin-Delgado.
\newblock Quantum google in a complex network.
\newblock {\em Scientific reports}, Vol.~3, , 2013.



\bibitem{HKSS2}
Yu.~Higuchi, N.~Konno, I.~Sato and E.~Segawa.
\newblock Spectral and asymptotic properties of {G}rover walks on crystal
  lattices.
\newblock {\em J. Funct. Anal.}, Vol. 267, No.~11, pp. 4197--4235, 2014.

\bibitem{NC}
M.A.~Nielsen and I.L.~Chuang.
\newblock {\em Quantum computation and quantum information}.
\newblock Cambridge university press, 2010.

\bibitem{Meyer}
D.A.~Meyer.
\newblock From quantum cellular automata to quantum lattice gases.
\newblock {\em J. Statist. Phys.}, Vol.~85, No. 5-6, pp. 551--574, 1996.

\bibitem{TFMK}
B.~Tregenna, W.~Flanagan, R.~Maile and V.~Kendon.
\newblock Controlling discrete quantum walks: coins and initial states.
\newblock {\em New Journal of Physics}, Vol.~5, No.~1, p.~83, 2003.

\bibitem{MS}
J.W.~Milnor and J.D.~Stasheff.
\newblock {\em Characteristic classes}.
\newblock Princeton University Press, Princeton, N. J.; University of Tokyo
  Press, Tokyo, 1974.
\newblock Annals of Mathematics Studies, No. 76.

\bibitem{MOS}
K.~Matsue, O.~Ogurisu and E.~Segawa.
\newblock Quantum walks on cubical sets : Construction and asymptotic behavior
  on $\mathbb{R}^2$, in preparation.

\bibitem{KMM}
T.~Kaczynski, K.~Mischaikow and M.~Mrozek.
\newblock {\em Computational homology}, Vol. 157 of {\em Applied Mathematical
  Sciences}.
\newblock Springer-Verlag, New York, 2004.


\end{thebibliography}

\appendix
\section{Simplicial complexes and homology groups}
\label{section-complexes}

We state a quick review of homology of simplicial complexes for readers who are not familiar with it. See e.g. \cite{KMM} for details.

\begin{dfn}\rm
Let $\mathbb{R}^N$ be a Euclidean space and $O_N$ be the origin of $\mathbb{R}^N$. Let $a_0,a_1,\cdots, a_n \in \mathbb{R}^N$ be points so that $n$ vectors $\{\overrightarrow{a_0a_i}\}_{i=1}^n$ are linearly independent.
An {\em $n$-simplex} is a set $|\sigma| \subset \mathbb{R}^N$ given by
\begin{equation*}
|\sigma| = \left\{\sum_{i=1}^n \lambda_i \overrightarrow{a_0a_i}\mid \lambda_i \geq 0,\ \sum_{i=1}^n \lambda_i = 1\right\}.
\end{equation*}
We also write $|\sigma|$ as $|a_0a_1 \cdots a_n|$ if we write the dependence of points $\{a_i\}_{i=0}^n$ explicitly.
A $k$-{\em face} of an $n$-simplex $|\sigma| = |a_0a_1\cdots a_n|$ is a $k$-simplex $|\tau|$ generated by $k$ points in $\{a_0\}_{i=0}^n$. In such a case, $|\sigma|$ is called {\em a coface} of $|\tau|$.
An $(n-1)$-face of an $n$-simplex $\sigma$ is often called {\em a primary face} of $\sigma$.
\end{dfn}
For example, for a given simplex $|\sigma| = |abc|$, edges $|ab|$, $|bc|$ and $|ca|$ are primary faces of $\sigma$. Also, vertices $|a|$, $|b|$ and $|c|$ are $0$-faces of $|\sigma|$. Finally, $|\sigma|$ is a coface of $|a|$, $|b|$, $|c|$, $|ab|$, $|bc|$ and $|ca|$.

\begin{dfn}\rm
A {\em simplicial complex} $\mathcal{K}$ is the collection of simplices satisfying
\begin{itemize}
\item If $|\sigma| \in \mathcal{K}$, then all faces of $|\sigma|$ are also elements in $\mathcal{K}$.
\item If $|\sigma_1|, |\sigma_2| \in \mathcal{K}$ and if $|\sigma_1|\cap |\sigma_2| \not = \emptyset$, then $|\sigma_1|\cap |\sigma_2|$ is a face of both $|\sigma_1|$ and $|\sigma_2|$. 
\end{itemize}
For a given simplicial complex $\mathcal{K}$, the union of all simplices of $\mathcal{K}$ is the {\em polytope} of $\mathcal{K}$ and is  denoted by $\mathcal{K}$. A set $P$ is a {\em polyhedron} if it is the polytope $\mathcal{K}$ of a simplicial complex $\mathcal{K}$.
\end{dfn}

Let $\mathcal{K} = \{K_k\}_{k\geq 0}$ be a simplicial complex, where $K_k = \{|\sigma| \in \mathcal{K} \mid |\sigma| \text{ is a $k$-simplex}\}$. If $n = \max \{k \mid K_k \not = \emptyset\} < \infty$, then we call $\mathcal{K}$ an {\em $n$-dimensional} simplicial complex. 
For a simplicial complex $\mathcal{K}$ with $\dim \mathcal{K}= n$, its {\em $m$-skeleton} is defined by $\mathcal{K}^{(m)} := \{K_k\}_{k\geq 0}^m$ for $m\leq n$.
For each $|\sigma| \in K_k$, we call the number $k$ {\em the dimension} of $|\sigma|$. Simplicial complexes admit several classes to be considered.

\begin{dfn}\rm
\label{dfn-strong-conn}
For a simplicial complex $\mathcal{K}$, a {\em facet} in $\mathcal{K}$ is a simplex $\sigma \in \mathcal{K}$ which is maximal with respect to the inclusion relation of sets. A simplicial complex $\mathcal{K}$ is {\em pure} if all facets in $\mathcal{K}$ have an identical dimension. An $n$-dimensional pure simplicial complex $\mathcal{K}$ is {\em strongly connected} if, for each $|\sigma|, |\tau| \in K_n$, there is a sequence of $n$-simplices $\{|\sigma_j|\}_{j=0}^k$ with $|\sigma_0| = |\sigma|$ and $|\sigma_k| = |\tau|$ such that $|\sigma_{j-1}| \cap |\sigma_j|$ is a primary face of $|\sigma_{j-1}|$ and $|\sigma_j|$ for $j=1,\cdots, n$. 
\end{dfn}

\begin{rem}\rm
We often call an $(n-1)$-face of an $n$-simplex $|\sigma|$ {\em a facet of $|\sigma|$}. It is completely different from facets in simplicial complexes.
\end{rem}

The core of homology is to translate geometric objects to algebraic ones in terms of chains.
\begin{dfn}\rm
Let $\mathcal{K}$ be a simplicial complex. For each $|\sigma| = |a_0 a_1 \cdots a_k|\in K_k$, the associated {\em $k$-chain} is the function $\widehat{|\sigma|} : K_k \to \{0,1\}$ given by
\begin{equation*}
\widehat{|\sigma|}(|\sigma'|) := \begin{cases}
	1 & \text{ if $\sigma' = \sigma$}\\
	0 & \text{ otherwise}
	\end{cases}
\end{equation*}
with the following rule: for a permutation $\pi\in S_{k+1}$, identify $\widehat{|\sigma|}(\pi |\sigma|)$ with $(\det \pi)\widehat{|\sigma|}(|\sigma|)$. 
\end{dfn}

Define 
\begin{equation*}
C_k(\mathcal{K}):= \left\{ \sum_{j=1}^k a_j \widehat{|\sigma_j|}\mid a_j \in \mathbb{Z},\ |\sigma_j| \in K_k\right\}.
\end{equation*}
This is called {\em $k$-th chain group} of $\mathcal{K}$, which is a $\mathbb{Z}$-module.

\begin{dfn}\rm
For $|\sigma| = |a_0 a_1 \cdots a_k|\in K_k$, define the {\em boundary} $\partial \widehat{|\sigma|}$ of the chain $\widehat{|\sigma|}$ by
\begin{equation*}
\partial_k \widehat{|\sigma|} := \sum_{j=0}^k (-1)^j \widehat{|\cdots a_{j-1}a_{j+1}\cdots |},
\end{equation*}
which is an element in $C_{k-1}$. Extending linearly this definition, we obtain a linear map $\partial_k : C_k(\mathcal{K})\to C_{k-1}(\mathcal{K})$.
This map is called {\em ($k$-th) boundary map} of $\mathcal{K}$. 
\end{dfn}
The important property of boundary maps is the following.
\begin{prop}
\label{prop-complex}
$\partial_k \circ \partial_{k+1} = 0: C_{k+1}(\mathcal{K})\to C_{k-1}(\mathcal{K})$.
\end{prop}

A pair $\{C_k(\mathcal{K}), \partial_k\}_{k\in \mathbb{Z}}$ consisting of sequences of chain groups and boundary maps is called a {\em chain complex} of $\mathcal{K}$.
Then we are ready to define homology groups.
\begin{dfn}\rm
For a simplicial complex $\mathcal{K}$, we define
\begin{equation*}
Z_k(\mathcal{K}):= {\rm Ker}\partial_k,\quad  B_k(\mathcal{K}):= {\rm Im}\partial_{k+1}.
\end{equation*}
Both are submodules of $C_k(\mathcal{K})$. An element in $Z_k(\mathcal{K})$ is called a {\em $k$-cycle} and an element in $B_k(\mathcal{K})$ is called a {\em $k$-boundary}. Thanks to Proposition \ref{prop-complex}, $B_k(\mathcal{K})$ is a submodule of $Z_k(\mathcal{K})$. Thus the quotient module
\begin{equation*}
H_k(\mathcal{K}):= Z_k(\mathcal{K}) / B_k(\mathcal{K})
\end{equation*}
can be considered. This quotient group is called {\em the $k$-th homology group of $\mathcal{K}$}. Roughly speaking, the $k$-th homology group describe the information of $k$-dimensional holes in $\mathcal{K}$.
\end{dfn}

The simple example of homology is shown in Figure \ref{fig-homology}.
\begin{figure}[htbp]\em
\begin{minipage}{0.5\hsize}
\centering
\includegraphics[width=3.0cm]{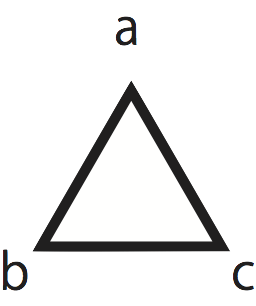}
(a)
\end{minipage}
\begin{minipage}{0.5\hsize}
\centering
\includegraphics[width=3.0cm]{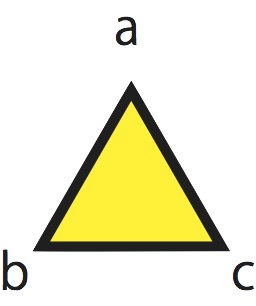}
(b)
\end{minipage}
\caption{The first homology class : a simple example.}
\label{fig-homology}
(a). A triangle without any $2$-simplices. In this case, the chain $\widehat{|ab|} + \widehat{|bc|} + \widehat{|ca|}$ becomes a $1$-cycle. Since no $2$-simplices exist, this $1$-cycle defines a generator of the first homology group, which indicates that this triangle has a hole.\\
(b). A triangle with a $2$-simplex $|abc|$. In this case, the chain $\widehat{|ab|} + \widehat{|bc|} + \widehat{|ca|}$ becomes a $1$-cycle. On the other hand, this $1$-cycle is also the boundary, which follows from $\widehat{|ab|} + \widehat{|bc|} + \widehat{|ca|} = \partial_2(\widehat{|abc|})$. This fact implies that this filled triangle has the trivial first homology group, in other words, the filled triangle does not have any holes.
\end{figure}

\section{Implementations: time evolution of the S-quantum walk on an infinite cylinder with a tetrahedron}
\label{section-implementation}

In this section we describe concrete implementations of time evolutions of S-quantum walk on $\mathcal{K}_2$, in particular, on and around tetrahedron drawn in Figure \ref{fig-tetrahedron}. On the other region, such evolutions follow from those on $\mathcal{K}_0$ discussed in Subsection \ref{section-setting-R2}. All labeling of simplices and bases on them are drawn in Figure \ref{fig-tetrahedron} and \ref{fig-state-tetrahedron}, respectively. Readers should refer to descriptions in these figures to study time evolutions stated below.

For the notation rule stated in Remark \ref{rem-notation}, we use bold-letter labelings shown in Figure \ref{fig-tetrahedron} instead of integer indices.
By following the definition of $U$, the evolution of base elements on the 2-simplex ${\bf x}\in \{ {\bf a},{\bf b},{\bf c}, {\bf A},{\bf B},{\bf C},{\bf D} \}$ is given in Subsections \ref{section-around-tetra} and \ref{section-on-tetra}.

\begin{figure}[htbp]\em
\begin{minipage}{1\hsize}
\centering
\includegraphics[width=12.0cm]{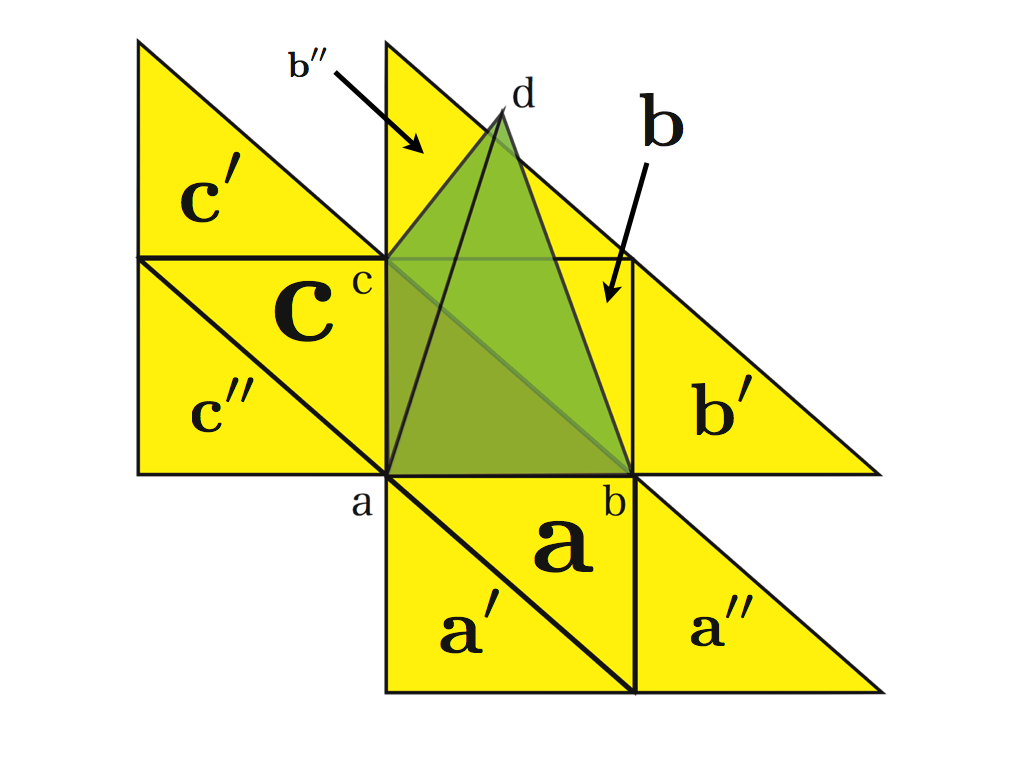}
\end{minipage}
\caption{Labeling of simplices around tetrahedron in $\mathcal{K}_2$}
\label{fig-tetrahedron}
A tetrahedron is put on the infinite cylinder $\mathcal{K}_1$. In this figure we put a tetrahedron, colored by green, on ${\bf A}\equiv |abc|$. As stated in Figure \ref{fig-state}-(a) the base element on ${\bf A}$ is determined by
\begin{equation*}
\delta_{{\bf A},0} = \delta^{(2)}_{[abc]},\ \delta_{{\bf A},1} = \delta^{(2)}_{[bca]},\ \delta_{{\bf A},2} = \delta^{(2)}_{[cab]},\ \delta_{{\bf A},3} = \delta^{(2)}_{[acb]},\ \delta_{{\bf A},4} = \delta^{(2)}_{[cba]},\ \delta_{{\bf A},5} = \delta^{(2)}_{[bac]}.
\end{equation*}
Simplices on the tetrahedron are labeled by ${\bf B}\equiv |abd|$, ${\bf C}\equiv |bcd|$ and ${\bf D}\equiv |cad|$. See also Figure \ref{fig-state-tetrahedron}. Simplices which have effects from the tetrahedron are labeled by {\bf a}, {\bf b} and  {\bf c}.
\end{figure}

\begin{figure}[htbp]\em
\begin{minipage}{1\hsize}
\centering
\includegraphics[width=12.0cm]{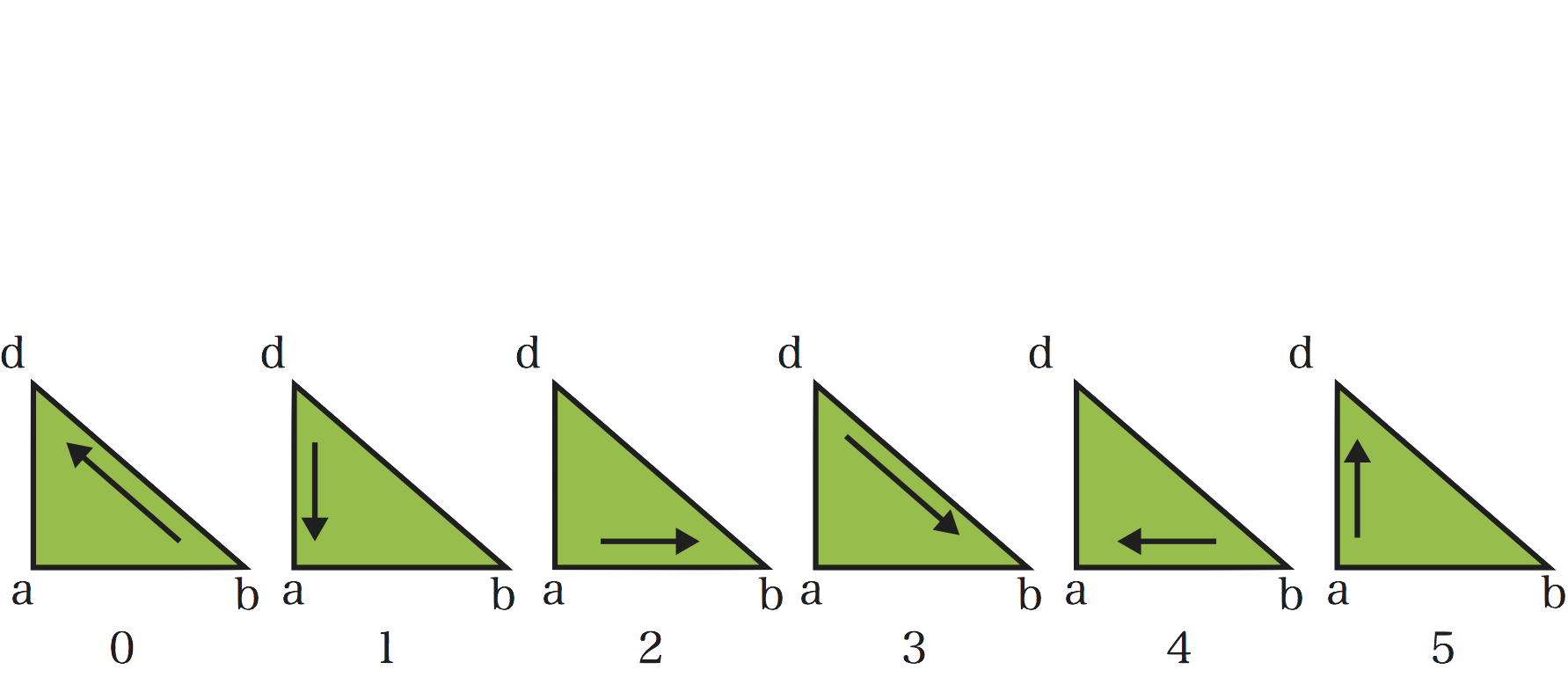}
(a)
\end{minipage}
\begin{minipage}{1\hsize}
\centering
\includegraphics[width=12.0cm]{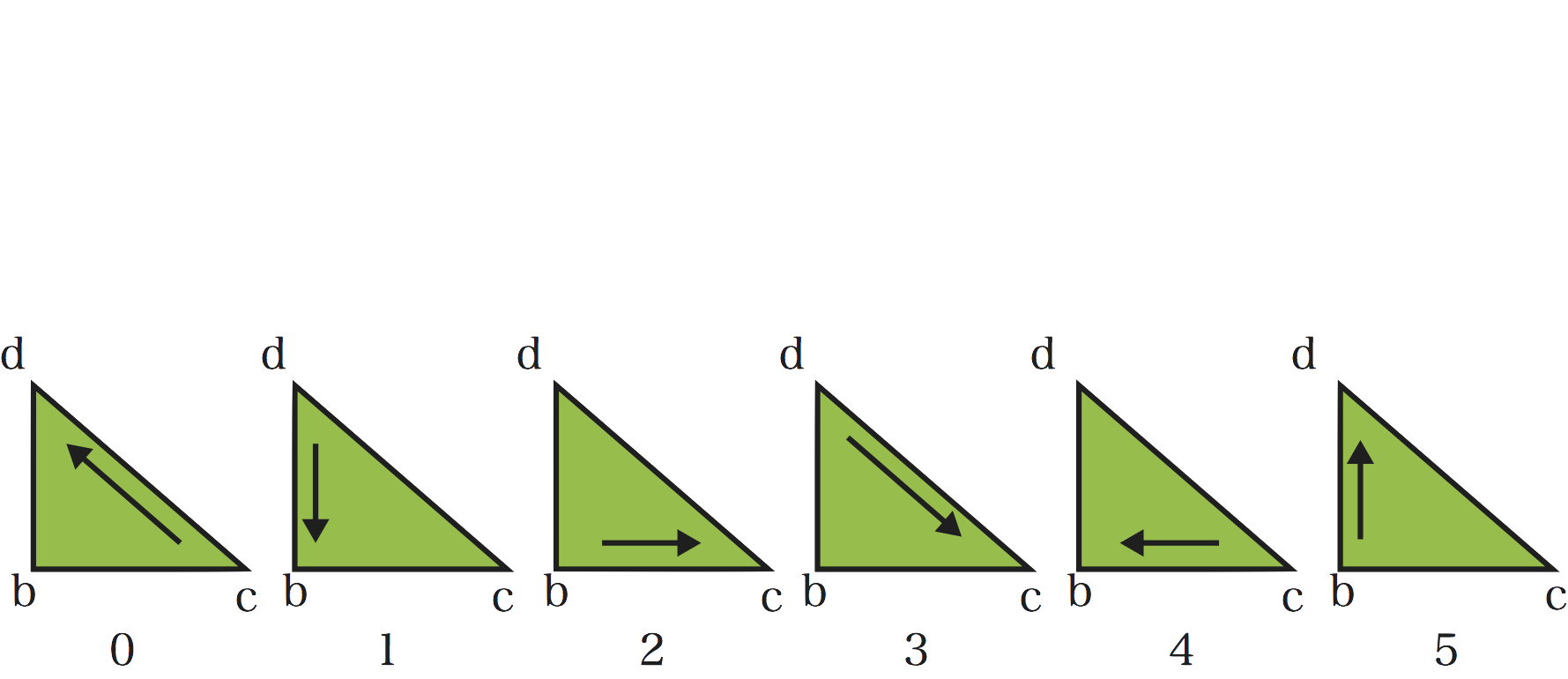}
(b)
\end{minipage}
\begin{minipage}{1\hsize}
\centering
\includegraphics[width=12.0cm]{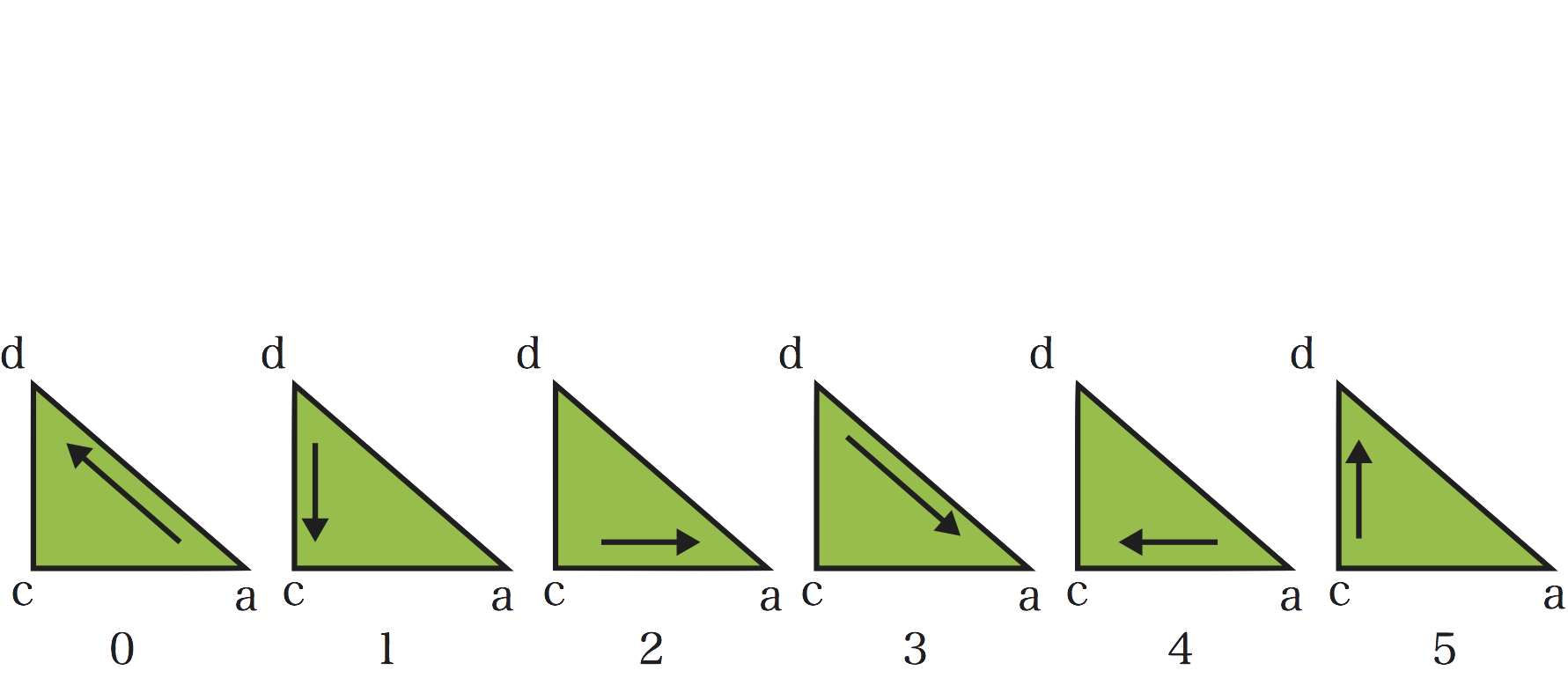}
(c)
\end{minipage}
\caption{Labeling of bases on the tetrahedron}
\label{fig-state-tetrahedron}
Figures here denote the graphical representations of the following bases for $l\in \{0,1,2,3,4,5\}$. Labelings of  {\bf B}, {\bf C} and {\bf D} are followed by Figure \ref{fig-tetrahedron}.
(a): $\delta_{{\bf B}, l}$. (b): $\delta_{{\bf C}, l}$. (c): $\delta_{{\bf D}, l}$. 
\end{figure}

\subsection{$2$-simplices around tetrahedron}
\label{section-around-tetra}
\begin{description}
\item[Case 1: $2$-simplex labeled {\bf a} (Figure \ref{fig-tetrahedron} and \ref{fig-state-tetrahedron})] 
\end{description}
\begin{align*}
\begin{pmatrix}
\delta_{{\bf a},0}\\
\delta_{{\bf a},1}\\
\delta_{{\bf a},2}\\
\delta_{{\bf a},3}\\
\delta_{{\bf a},4}\\
\delta_{{\bf a},5}
\end{pmatrix}
&\mapsto
\begin{pmatrix}
(2|w_{{\bf a},0}|^2-1)\cdot \delta_{{\bf a},1} \\
(2|w_{{\bf a},1}|^2-1)\cdot \delta_{{\bf a},2} \\
(2|w_{{\bf a},2}|^2-1)\cdot \delta_{{\bf a},0} \\
(2|w_{{\bf a},3}|^2-1)\cdot \delta_{{\bf a},4} \\
(2|w_{{\bf a},4}|^2-1)\cdot \delta_{{\bf a},5} \\
(2|w_{{\bf a},5}|^2-1)\cdot \delta_{{\bf a},3}
\end{pmatrix}
+
\begin{pmatrix}
2\overline{w_{{\bf a},0}}w_{{\bf A},4} \cdot \delta_{{\bf A},5}\\
2\overline{w_{{\bf a},1}}w_{{\bf a}',3} \cdot \delta_{{\bf a}',4}\\
2\overline{w_{{\bf a},2}}w_{{\bf a}'',5} \cdot \delta_{{\bf a}'',3}\\
2\overline{w_{{\bf a},3}}w_{{\bf A},2} \cdot \delta_{{\bf A},0}\\
2\overline{w_{{\bf a},4}}w_{{\bf a}'',1} \cdot \delta_{{\bf a}'',2}\\
2\overline{w_{{\bf a},5}}w_{{\bf a}',0} \cdot \delta_{{\bf a}',1}
\end{pmatrix}
+
\begin{pmatrix}
2\overline{w_{{\bf a},0}}w_{{\bf B},4} \cdot \delta_{{\bf B},5}\\
0\\
0\\
2\overline{w_{{\bf a},3}}w_{{\bf B},2} \cdot \delta_{{\bf B},0}\\
0\\
0
\end{pmatrix}.
\end{align*}

\begin{description}
\item[Case 2: $2$-simplex labeled {\bf b} (Figure \ref{fig-tetrahedron} and  \ref{fig-state-tetrahedron})] 
\end{description}
\begin{align*}
\begin{pmatrix}
\delta_{{\bf b},0}\\
\delta_{{\bf b},1}\\
\delta_{{\bf b},2}\\
\delta_{{\bf b},3}\\
\delta_{{\bf b},4}\\
\delta_{{\bf b},5}
\end{pmatrix}
\mapsto
\begin{pmatrix}
(2|w_{{\bf b},0}|^2-1)\cdot \delta_{{\bf b},1} \\
(2|w_{{\bf b},1}|^2-1)\cdot \delta_{{\bf b},2} \\
(2|w_{{\bf b},2}|^2-1)\cdot \delta_{{\bf b},0} \\
(2|w_{{\bf b},3}|^2-1)\cdot \delta_{{\bf b},4} \\
(2|w_{{\bf b},4}|^2-1)\cdot \delta_{{\bf b},5} \\
(2|w_{{\bf b},5}|^2-1)\cdot \delta_{{\bf b},3}
\end{pmatrix}
+
\begin{pmatrix}
2\overline{w_{{\bf b},0}}w_{{\bf b}'',4} \cdot \delta_{{\bf b}'',5}\\
2\overline{w_{{\bf b},1}}w_{{\bf A},3} \cdot \delta_{{\bf A},4}\\
2\overline{w_{{\bf b},2}}w_{{\bf b}',5} \cdot \delta_{{\bf b}',3}\\
2\overline{w_{{\bf b},3}}w_{{\bf b}'',2} \cdot \delta_{{\bf b}'',0}\\
2\overline{w_{{\bf b},4}}w_{{\bf b}',1} \cdot \delta_{{\bf b}',2}\\
2\overline{w_{{\bf b},5}}w_{{\bf A},0} \cdot \delta_{{\bf A},1}
\end{pmatrix}
+
\begin{pmatrix}
0\\
2\overline{w_{{\bf b},1}}w_{{\bf C},4} \cdot \delta_{{\bf C},5}\\
0\\
0\\
0\\
2\overline{w_{{\bf b},5}}w_{{\bf C},2} \cdot \delta_{{\bf C},0}
\end{pmatrix}.
\end{align*}

\begin{description}
\item[Case 3: $2$-simplex labeled {\bf c} (Figure \ref{fig-tetrahedron} and  \ref{fig-state-tetrahedron})] 
\end{description}
\begin{align*}
\begin{pmatrix}
\delta_{{\bf c},0}\\
\delta_{{\bf c},1}\\
\delta_{{\bf c},2}\\
\delta_{{\bf c},3}\\
\delta_{{\bf c},4}\\
\delta_{{\bf c},5}
\end{pmatrix}
&\mapsto
\begin{pmatrix}
(2|w_{{\bf c},0}|^2-1)\cdot \delta_{{\bf c},1} \\
(2|w_{{\bf c},1}|^2-1)\cdot \delta_{{\bf c},2} \\
(2|w_{{\bf c},2}|^2-1)\cdot \delta_{{\bf c},0} \\
(2|w_{{\bf c},3}|^2-1)\cdot \delta_{{\bf c},4} \\
(2|w_{{\bf c},4}|^2-1)\cdot \delta_{{\bf c},5} \\
(2|w_{{\bf c},5}|^2-1)\cdot \delta_{{\bf c},3}
\end{pmatrix}
+
\begin{pmatrix}
2\overline{w_{{\bf c},0}}w_{{\bf c}',4} \cdot \delta_{{\bf c}',5}\\
2\overline{w_{{\bf c},1}}w_{{\bf c}'',3} \cdot \delta_{{\bf c}'',4}\\
2\overline{w_{{\bf c},2}}w_{{\bf A},5} \cdot \delta_{{\bf A},3}\\
2\overline{w_{{\bf c},3}}w_{{\bf c}',2} \cdot \delta_{{\bf c}',0}\\
2\overline{w_{{\bf c},4}}w_{{\bf A},1} \cdot \delta_{{\bf A},2}\\
2\overline{w_{{\bf c},5}}w_{{\bf c}'',0} \cdot \delta_{{\bf c}'',1}
\end{pmatrix}
+
\begin{pmatrix}
0\\
0\\
2\overline{w_{{\bf c},2}}w_{{\bf D},4} \cdot \delta_{{\bf D},5}\\
0\\
2\overline{w_{{\bf c},4}}w_{{\bf D},2} \cdot \delta_{{\bf D},0}\\
0
\end{pmatrix}.
\end{align*}

\subsection{$2$-simplices on tetrahedron}
\label{section-on-tetra}
\begin{description}
\item[Case 1: $2$-simplex labeled {\bf A} (Figure \ref{fig-tetrahedron} and  \ref{fig-state-tetrahedron})] 
\end{description}
\begin{align*}
\begin{pmatrix}
\delta_{{\bf A},0}\\
\delta_{{\bf A},1}\\
\delta_{{\bf A},2}\\
\delta_{{\bf A},3}\\
\delta_{{\bf A},4}\\
\delta_{{\bf A},5}
\end{pmatrix}
&\mapsto
\begin{pmatrix}
(2|w_{{\bf A},0}|^2-1)\cdot \delta_{{\bf A},1} \\
(2|w_{{\bf A},1}|^2-1)\cdot \delta_{{\bf A},2} \\
(2|w_{{\bf A},2}|^2-1)\cdot \delta_{{\bf A},0} \\
(2|w_{{\bf A},3}|^2-1)\cdot \delta_{{\bf A},4} \\
(2|w_{{\bf A},4}|^2-1)\cdot \delta_{{\bf A},5} \\
(2|w_{{\bf A},5}|^2-1)\cdot \delta_{{\bf A},3}
\end{pmatrix}
 +
\begin{pmatrix}
2\overline{w_{{\bf A},0}}w_{{\bf b},5} \cdot \delta_{{\bf b},3}\\
2\overline{w_{{\bf A},1}}w_{{\bf c},4} \cdot \delta_{{\bf c},5}\\
2\overline{w_{{\bf A},2}}w_{{\bf a},3} \cdot \delta_{{\bf a},4}\\
2\overline{w_{{\bf A},3}}w_{{\bf b},1} \cdot \delta_{{\bf b},2}\\
2\overline{w_{{\bf A},4}}w_{{\bf a},0} \cdot \delta_{{\bf a},1}\\
2\overline{w_{{\bf A},5}}w_{{\bf c},2} \cdot \delta_{{\bf c},0}
\end{pmatrix}
+
\begin{pmatrix}
2\overline{w_{{\bf A},0}}w_{{\bf C},2} \cdot \delta_{{\bf C},0}\\
2\overline{w_{{\bf A},1}}w_{{\bf D},2} \cdot \delta_{{\bf D},0}\\
2\overline{w_{{\bf A},2}}w_{{\bf B},2} \cdot \delta_{{\bf B},0}\\
2\overline{w_{{\bf A},3}}w_{{\bf C},4} \cdot \delta_{{\bf C},5}\\
2\overline{w_{{\bf A},4}}w_{{\bf B},4} \cdot \delta_{{\bf B},5}\\
2\overline{w_{{\bf A},5}}w_{{\bf D},4} \cdot \delta_{{\bf D},5}
\end{pmatrix}.
\end{align*}

\begin{description}
\item[Case 2: $2$-simplex labeled {\bf B} (Figure \ref{fig-tetrahedron} and  \ref{fig-state-tetrahedron})] 
\end{description}
\begin{align*}
\begin{pmatrix}
\delta_{{\bf B},0}\\
\delta_{{\bf B},1}\\
\delta_{{\bf B},2}\\
\delta_{{\bf B},3}\\
\delta_{{\bf B},4}\\
\delta_{{\bf B},5}
\end{pmatrix}
&\mapsto
\begin{pmatrix}
(2|w_{{\bf B},0}|^2-1)\cdot \delta_{{\bf B},1} \\
(2|w_{{\bf B},1}|^2-1)\cdot \delta_{{\bf B},2} \\
(2|w_{{\bf B},2}|^2-1)\cdot \delta_{{\bf B},0} \\
(2|w_{{\bf B},3}|^2-1)\cdot \delta_{{\bf B},4} \\
(2|w_{{\bf B},4}|^2-1)\cdot \delta_{{\bf B},5} \\
(2|w_{{\bf B},5}|^2-1)\cdot \delta_{{\bf B},3}
\end{pmatrix}
 +
\begin{pmatrix}
0\\
0\\
2\overline{w_{{\bf B},2}}w_{{\bf a},6} \cdot \delta_{{\bf a},4} \\
0\\
2\overline{w_{{\bf B},4}}w_{{\bf a},0} \cdot \delta_{{\bf a},1} \\
0
\end{pmatrix}
+
\begin{pmatrix}
2\overline{w_{{\bf B},0}}w_{{\bf C},5} \cdot \delta_{{\bf C},3}\\
2\overline{w_{{\bf B},1}}w_{{\bf D},3} \cdot \delta_{{\bf D},4}\\
2\overline{w_{{\bf B},2}}w_{{\bf A},2} \cdot \delta_{{\bf A},0}\\
2\overline{w_{{\bf B},3}}w_{{\bf C},1} \cdot \delta_{{\bf C},2}\\
2\overline{w_{{\bf B},4}}w_{{\bf A},4} \cdot \delta_{{\bf A},5}\\
2\overline{w_{{\bf B},5}}w_{{\bf D},0} \cdot \delta_{{\bf D},1}
\end{pmatrix}.
\end{align*}

\begin{description}
\item[Case 3: $2$-simplex labeled {\bf C} (Figure \ref{fig-tetrahedron} and  \ref{fig-state-tetrahedron})] 
\end{description}
\begin{align*}
\begin{pmatrix}
\delta_{{\bf C},0}\\
\delta_{{\bf C},1}\\
\delta_{{\bf C},2}\\
\delta_{{\bf C},3}\\
\delta_{{\bf C},4}\\
\delta_{{\bf C},5}
\end{pmatrix}
&\mapsto
\begin{pmatrix}
(2|w_{{\bf C},0}|^2-1)\cdot \delta_{{\bf C},1} \\
(2|w_{{\bf C},1}|^2-1)\cdot \delta_{{\bf C},2} \\
(2|w_{{\bf C},2}|^2-1)\cdot \delta_{{\bf C},0} \\
(2|w_{{\bf C},3}|^2-1)\cdot \delta_{{\bf C},4} \\
(2|w_{{\bf C},4}|^2-1)\cdot \delta_{{\bf C},5} \\
(2|w_{{\bf C},5}|^2-1)\cdot \delta_{{\bf C},3}
\end{pmatrix}
 +
\begin{pmatrix}
0\\
0\\
2\overline{w_{{\bf C},2}}w_{{\bf b},5} \cdot \delta_{{\bf b},3} \\
0\\
2\overline{w_{{\bf C},4}}w_{{\bf b},1} \cdot \delta_{{\bf b},2} \\
0
\end{pmatrix}
+
\begin{pmatrix}
2\overline{w_{{\bf C},0}}w_{{\bf D},5} \cdot \delta_{{\bf D},3}\\
2\overline{w_{{\bf C},1}}w_{{\bf B},3} \cdot \delta_{{\bf B},4}\\
2\overline{w_{{\bf C},2}}w_{{\bf A},0} \cdot \delta_{{\bf A},1}\\
2\overline{w_{{\bf C},3}}w_{{\bf D},1} \cdot \delta_{{\bf D},2}\\
2\overline{w_{{\bf C},4}}w_{{\bf A},3} \cdot \delta_{{\bf A},4}\\
2\overline{w_{{\bf C},5}}w_{{\bf B},0} \cdot \delta_{{\bf B},1}
\end{pmatrix}.
\end{align*}

\begin{description}
\item[Case 4: $2$-simplex labeled {\bf D} (Figure \ref{fig-tetrahedron} and  \ref{fig-state-tetrahedron})] 
\end{description}
\begin{align*}
\begin{pmatrix}
\delta_{{\bf D},0}\\
\delta_{{\bf D},1}\\
\delta_{{\bf D},2}\\
\delta_{{\bf D},3}\\
\delta_{{\bf D},4}\\
\delta_{{\bf D},5}
\end{pmatrix}
&\mapsto
\begin{pmatrix}
(2|w_{{\bf D},0}|^2-1)\cdot \delta_{{\bf D},1} \\
(2|w_{{\bf D},1}|^2-1)\cdot \delta_{{\bf D},2} \\
(2|w_{{\bf D},2}|^2-1)\cdot \delta_{{\bf D},0} \\
(2|w_{{\bf D},3}|^2-1)\cdot \delta_{{\bf D},4} \\
(2|w_{{\bf D},4}|^2-1)\cdot \delta_{{\bf D},5} \\
(2|w_{{\bf D},5}|^2-1)\cdot \delta_{{\bf D},3}
\end{pmatrix}
 +
\begin{pmatrix}
0\\
0\\
2\overline{w_{{\bf D},2}}w_{{\bf c},4} \cdot \delta_{{\bf c},5} \\
0\\
2\overline{w_{{\bf D},2}}w_{{\bf c},2} \cdot \delta_{{\bf c},0} \\
0
\end{pmatrix}
+
\begin{pmatrix}
2\overline{w_{{\bf D},0}}w_{{\bf B},5} \cdot \delta_{{\bf B},3}\\
2\overline{w_{{\bf D},1}}w_{{\bf C},3} \cdot \delta_{{\bf C},4}\\
2\overline{w_{{\bf D},2}}w_{{\bf A},1} \cdot \delta_{{\bf A},2}\\
2\overline{w_{{\bf D},3}}w_{{\bf B},1} \cdot \delta_{{\bf B},2}\\
2\overline{w_{{\bf D},2}}w_{{\bf A},5} \cdot \delta_{{\bf A},3}\\
2\overline{w_{{\bf D},5}}w_{{\bf C},0} \cdot \delta_{{\bf C},1}
\end{pmatrix}.
\end{align*}

\section{Figures: behavior of S-quantum walks on simplicial complexes}

\begin{figure}[htbp]\em
\begin{minipage}{0.5\hsize}
\centering
\includegraphics[width=8.0cm]{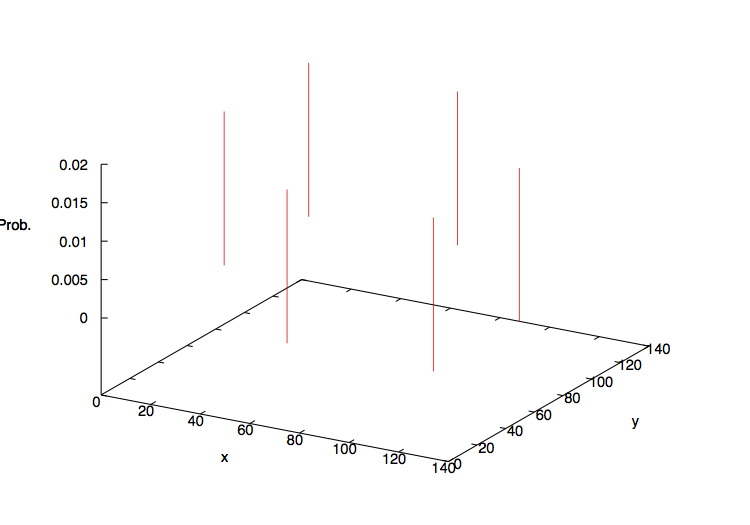}
(a)
\end{minipage}
\begin{minipage}{0.5\hsize}
\centering
\includegraphics[width=8.0cm]{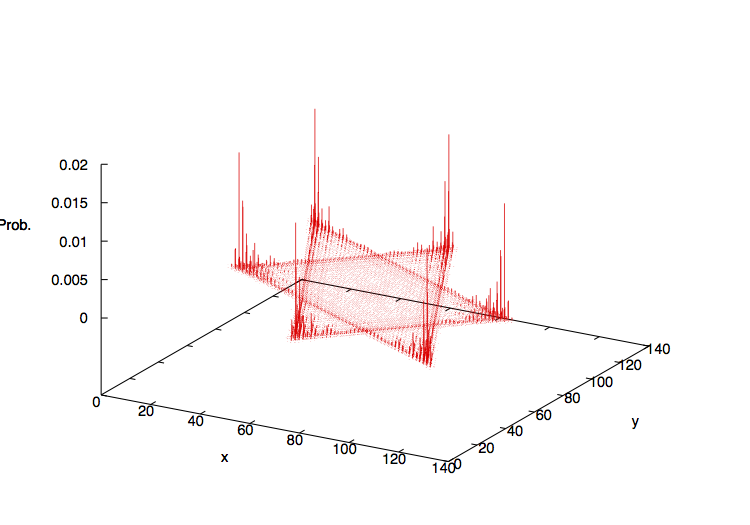}
(b)
\end{minipage}
\begin{minipage}{0.5\hsize}
\centering
\includegraphics[width=8.0cm]{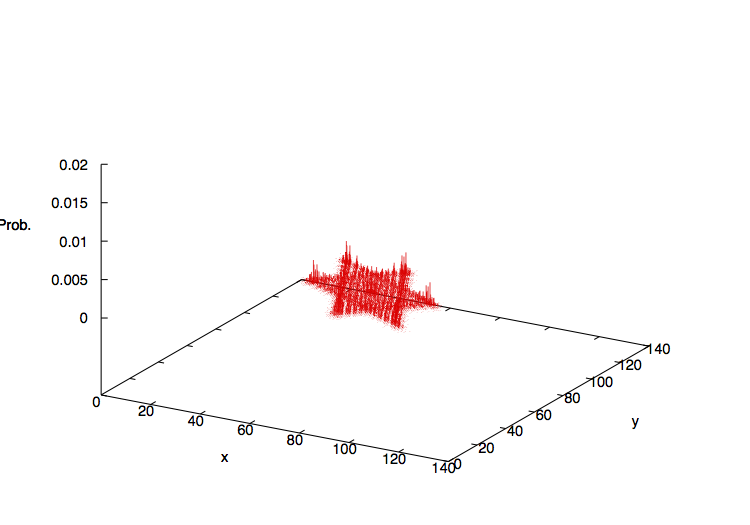}
(c)
\end{minipage}
\caption{S-quantum walk on $\mathcal{K}_0$ with differential weights : The Hexagram}
\label{fig-Swalk-R2}

\bigskip
(a). Probability density distribution after $120$ steps of the S-quantum walk with the weight (\ref{weight-R2-1}). Initial state at $|abc|$, $a=(0,0)\in \mathbb{R}^2$, $b=(1,0)\in \mathbb{R}^2$, $c=(0,1)\in \mathbb{R}^2$ moves in monotone directions. This is a typical example of the non-interactive quantum walk.

\bigskip
(b). Probability density distribution after $120$ steps of the S-quantum walk with the weight (\ref{weight-R2-2}). Initial state at $|abc|$ draws the hexagram. 

\bigskip
(c). Probability density distribution after $120$ steps of the S-quantum walk with the weight (\ref{weight-R2-3}). Initial state at $|abc|$ draws the hexagram but the speed of the spread is slower than the case (b). 

\bigskip
In all cases localization cannot be observed, which is the big difference from the quantum walk on crystal lattices \cite{HKSS2}.
\end{figure}

\begin{figure}[htbp]\em
\begin{minipage}{0.5\hsize}
\centering
\includegraphics[width=8.0cm]{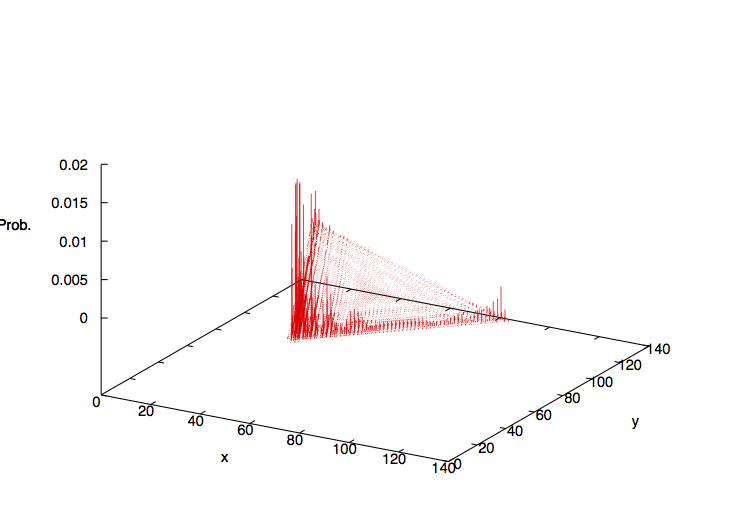}
(a)
\end{minipage}
\begin{minipage}{0,5\hsize}
\centering
\includegraphics[width=8.0cm]{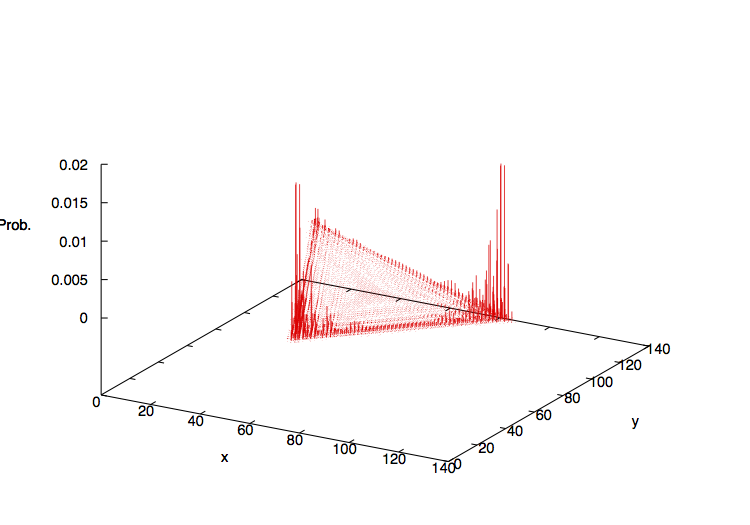}
(b)
\end{minipage}
\begin{minipage}{0,5\hsize}
\centering
\includegraphics[width=8.0cm]{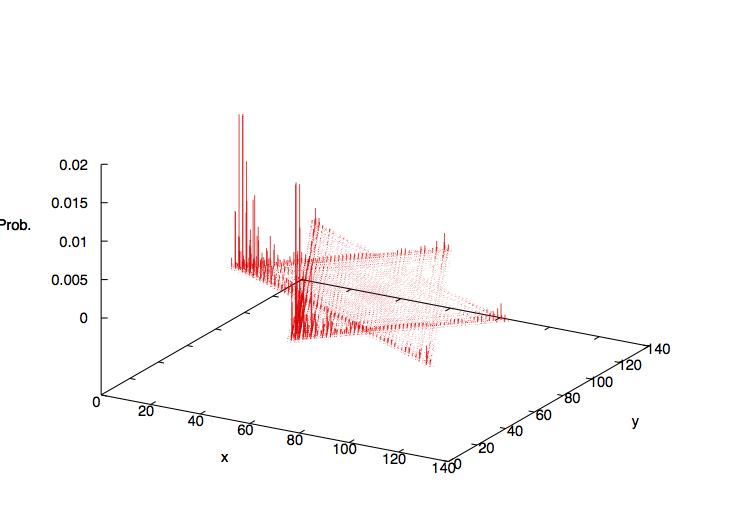}
(c)
\end{minipage}
\caption{S-quantum walk on $\mathcal{K}_0$ with different initial states : The Triangle}
\label{fig-Swalk-1state}
Probability density distribution after $120$ steps of the S-quantum walk with the weight (\ref{weight-R2-2}) and the following initial states.

\bigskip
(a). The initial state is $\Psi_0 = \delta^{(2)}_{[abc]}$, where $a=(0,0)\in \mathbb{R}^2$, $b=(1,0)\in \mathbb{R}^2$, $c=(0,1)\in \mathbb{R}^2$. The S-quantum walk monotone spreads drawing triangle with high probability in one direction.

\bigskip
(b). The initial state is $\Psi_0 = \frac{1}{\sqrt{2}}\delta^{(2)}_{[abc]} + \frac{1}{\sqrt{2}}\delta^{(2)}_{[bca]}$. The S-quantum walk spreads drawing triangle with high probability in two directions.

\bigskip
(c). The initial state is $\Psi_0 = \frac{1}{\sqrt{2}}\delta^{(2)}_{[abc]} + \frac{1}{\sqrt{2}}\delta^{(2)}_{[acb]}$. The S-quantum walk spreads drawing hexagram with high probability in two directions. 
\end{figure}

\begin{figure}[htbp]\em
\begin{minipage}{0.5\hsize}
\centering
\includegraphics[width=7.0cm]{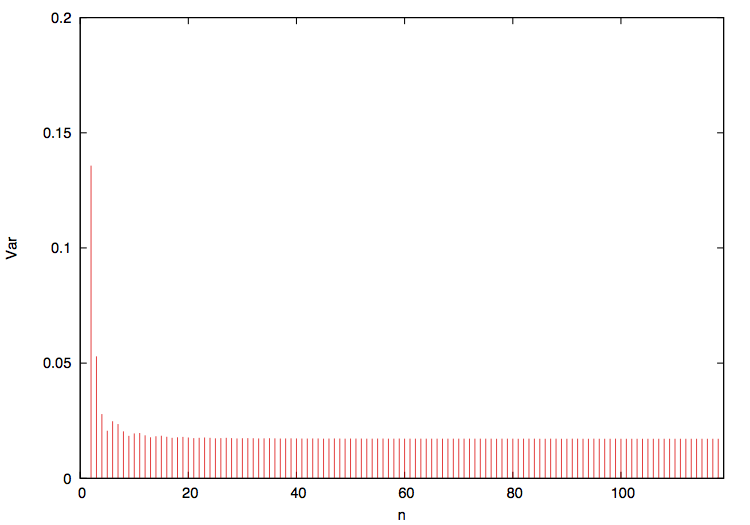}
(a)
\end{minipage}
\begin{minipage}{0.5\hsize}
\centering
\includegraphics[width=7.0cm]{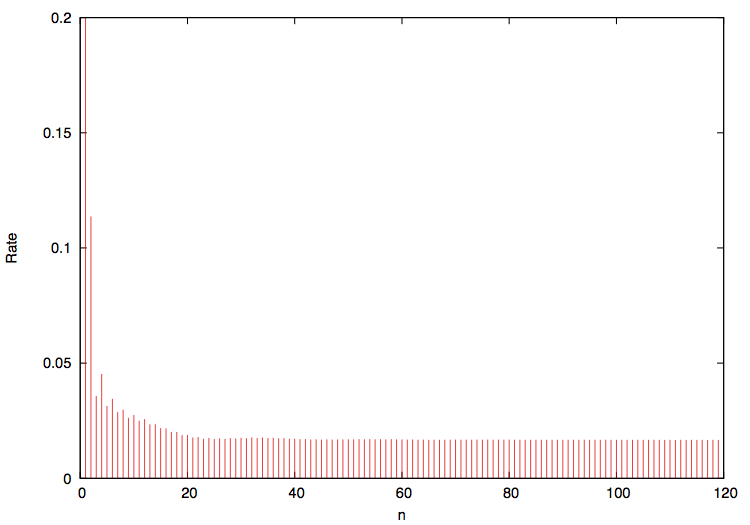}
(b)
\end{minipage}
\begin{minipage}{0.5\hsize}
\centering
\includegraphics[width=7.0cm]{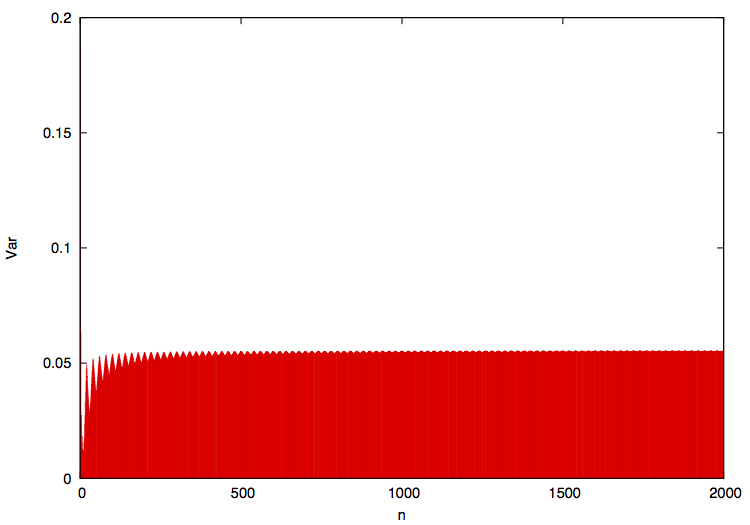}
(c)
\end{minipage}
\begin{minipage}{0.5\hsize}
\centering
\includegraphics[width=7.0cm]{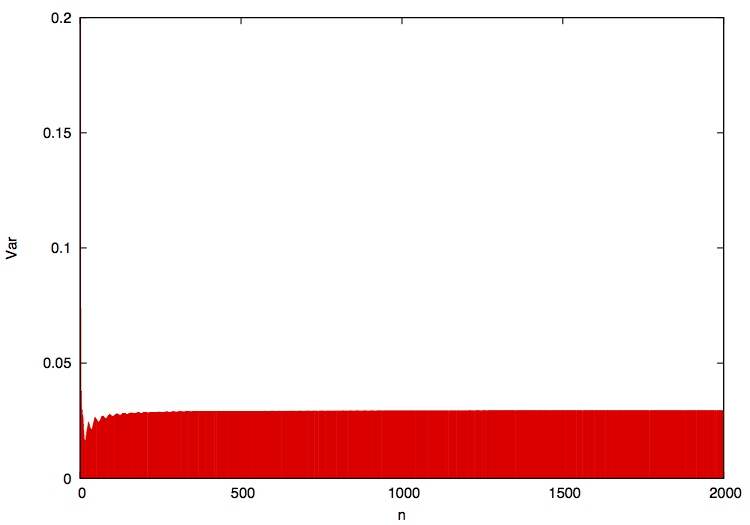}
(d)
\end{minipage}
\caption{The graph of variance $V_n/n^2$ in (\ref{variance-complex}): asymptotic behavior of the second moment of $\Psi_n/n$ with respect to the S-quantum walks $\Psi_n$. Row: time step $n$. Column: the value of $V_n/n^2$.}
\label{fig-spread}

\bigskip
(a). Time evolution of $V_n/n^2$ for the Hadamard walk on $\mathbb{Z}^2$ (e.g. \cite{K1}). The value of $V_n/n^2$ converges to $\approx 0.01719$.

\bigskip
(b). Time evolution of $V_n/n^2$ for the S-quantum walk on $\mathcal{K}_0$ with the weight (\ref{weight-R2-2}). The value of $V_n/n^2$ converges to $\approx 0.01671$.

\bigskip
(c). Time evolution of $V_n/n^2$ for the S-quantum walk on $\mathcal{K}_1$ with the weight (\ref{weight-R2-1}). The value of $V_n/n^2$ converges to $\approx 0.05545$.

\bigskip
(d). Time evolution of $V_n/n^2$ for the S-quantum walk on $\mathcal{K}_1$ with the weight (\ref{weight-R2-2}). The value of $V_n/n^2$ converges to $\approx 0.02941$.

\end{figure}

\begin{figure}[htbp]\em
\begin{minipage}{0.5\hsize}
\centering
\includegraphics[width=8.0cm]{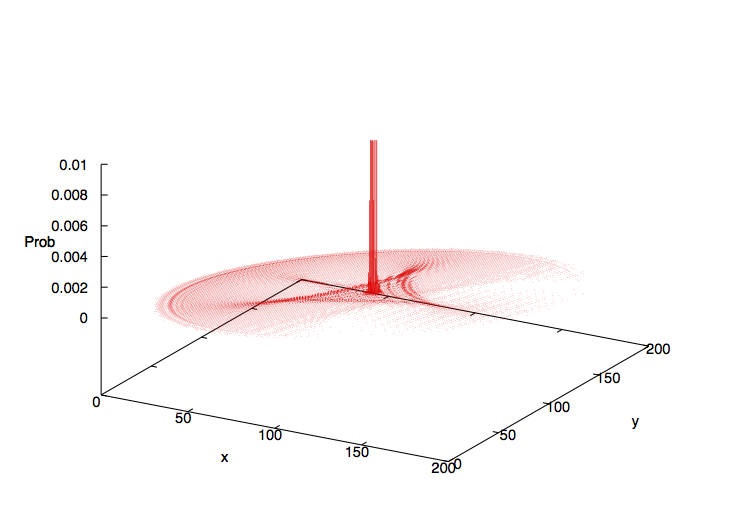}
(a)
\end{minipage}
\begin{minipage}{0.5\hsize}
\centering
\includegraphics[width=8.0cm]{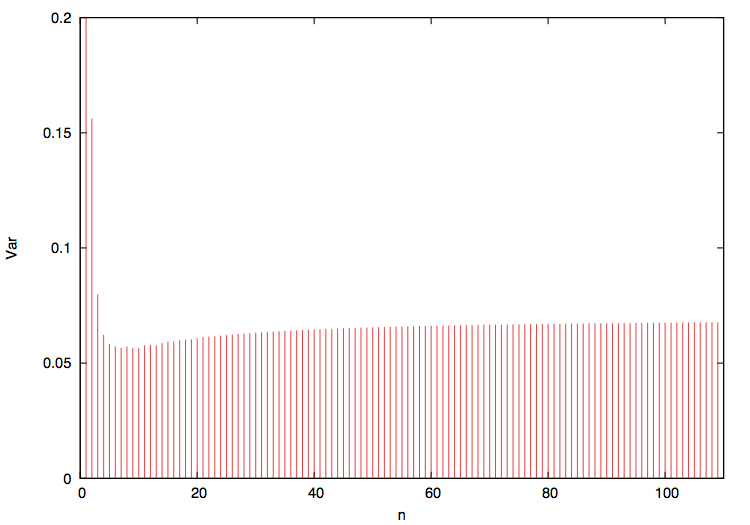}
(b)
\end{minipage}
\caption{The Grover walk $\Psi_n$ on the triangular lattice}
\label{fig-QW-triangle}

\bigskip
(a): Localization at the origin.\\
(b): Convergence of variance of the quantum walk $\Psi_n/n$. Row: time. Column: the value of the variance $V_n/n^2$ at the time $n$.
\end{figure}

\begin{figure}[htbp]\em
\begin{minipage}{0.5\hsize}
\centering
\includegraphics[width=8.0cm]{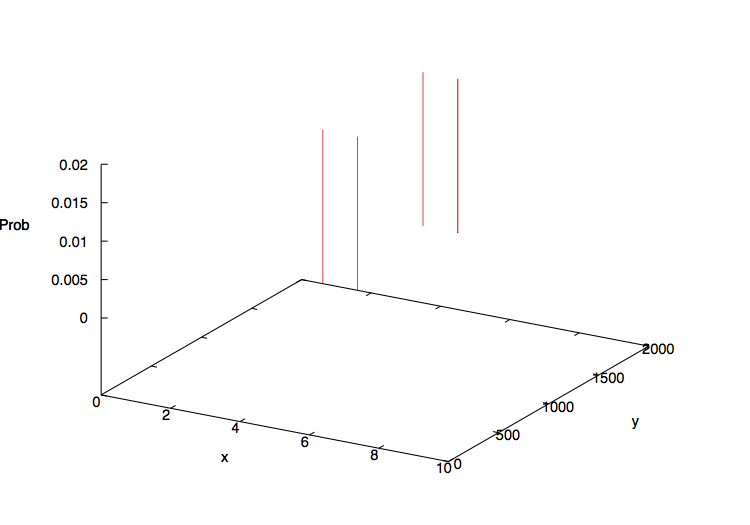}
(a)
\end{minipage}
\begin{minipage}{0.5\hsize}
\centering
\includegraphics[width=8.0cm]{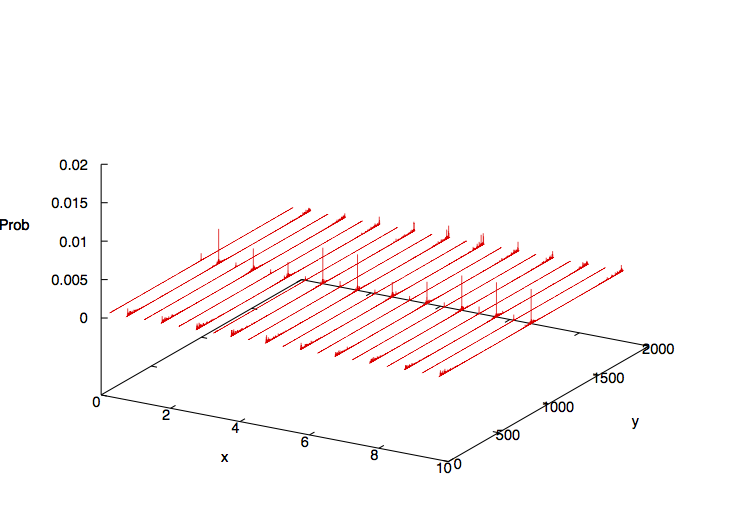}
(b)
\end{minipage}
\begin{minipage}{0.5\hsize}
\centering
\includegraphics[width=8.0cm]{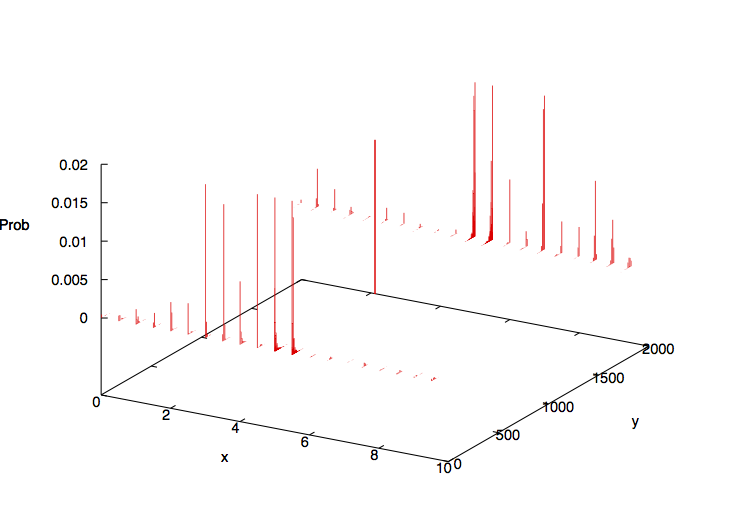}
(c)
\end{minipage}
\caption{S-quantum walk on $\mathcal{K}_1$ : (a), (b),  and on $\mathcal{K}_2$ : (c)}
\label{fig-Swalk-cylinder}
Figures (a) and (b) depict probability after 2000 steps of the S-quantum walk on $\mathcal{K}_1$ with the weights (\ref{weight-R2-1}) and (\ref{weight-R2-2}), respectively.
Figure (c) depict probability after 2000 steps of the S-quantum walk on $\mathcal{K}_2$ with the weight  (\ref{weight-Grover}). 
In all cases the initial state $\Psi_0$ is given by (\ref{initial-cylinder}) (see also Figure \ref{fig-tetrahedron}). 

\bigskip
(a). Initial state $\Psi_0$ moves in two infinite directions as well as staying on the cycle where the starting simplex is located with a positive probability. See also the graph \lq\lq cylinder-transmit" in Figure \ref{fig-time-average}.
This is an example of the non-interactive quantum walk.

\bigskip
(b). Initial state $\Psi_0$ moves in two infinite directions as well as staying on the cycle where the starting simplex is located with a positive probability, as in the case of (a). See also the graph \lq\lq cylinder-hetero" in Figure \ref{fig-time-average}.

\bigskip
(c). Localization occurs on the tetrahedron, while localization on the circle observed in the case of $\mathcal{K}_1$ is absorbed by the tetrahedron. See also the graphs  \lq\lq tetra-on" and \lq\lq tetra-off" in Figure \ref{fig-time-average}.

\end{figure}

\begin{figure}[htbp]\em
\begin{minipage}{0.5\hsize}
\centering
\includegraphics[width=8.0cm]{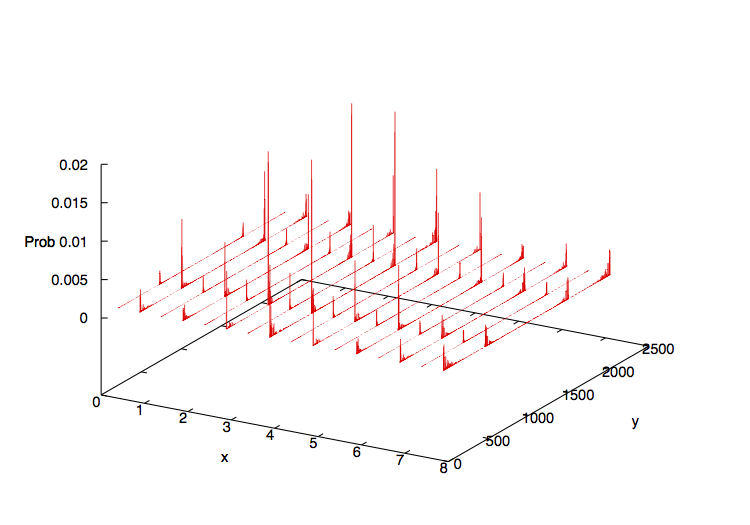}
(a)
\end{minipage}
\begin{minipage}{0.5\hsize}
\centering
\includegraphics[width=8.0cm]{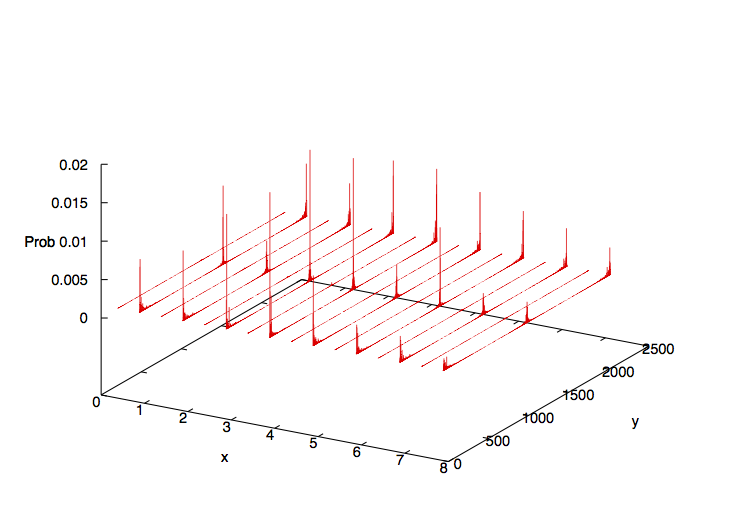}
(b)
\end{minipage}
\caption{S-quantum walks on $\mathcal{K}_3$ : (a), versus on $\mathcal{K}_1$ : (b)}
\label{fig-Mobius}

\bigskip
(a). Probability density distribution after $2100$ steps of the S-quantum walk on $\mathcal{K}_3$ with the weight (\ref{weight-Mobius}). Although the initial state exists only on the center of $\mathcal{K}_3$, quantum walker slides away from the center in $y$-coordinate. This implies that no localization can occur on $\mathcal{K}_3$, which reflects the geometry of $\mathcal{K}_3$ that the M\"{o}bius band is non-orientable.

\bigskip
(b). Probability density distribution after $2100$ steps of the S-quantum walk on $\mathcal{K}_1$ with the weight (\ref{weight-Mobius}). The initial state is the same as (a). In this case, localization can occur.
\end{figure}

\begin{figure}[htbp]\em
\begin{minipage}{1\hsize}
\centering
\includegraphics[width=10.0cm]{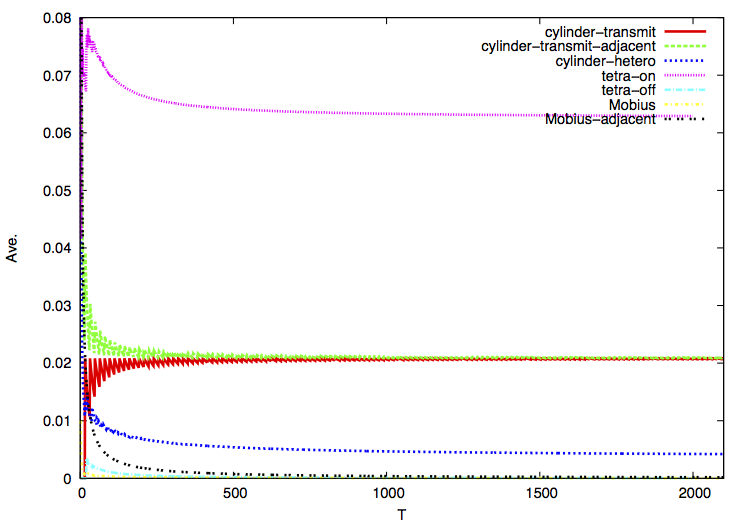}
\end{minipage}
\caption{Time-averaged probability of S-quantum walks on simplices}
\label{fig-time-average}
Each graph denotes the time-averaged probability $\bar \mu_T(|\sigma|)$ at time $T$ and at $|\sigma|\in K_2$. Let $|\sigma_0|$ be the support of the initial state $\Psi_0$ discussed in Section \ref{section-numerical}, labelled by $(i,j,k)=(N/2, N/2, 0)$. 
\begin{itemize}
\item \lq\lq cylinder-transmit" (red) denotes $\bar \mu_T^{(\Psi_0)}(|\sigma_0|)$ for S-quantum walk on $\mathcal{K}_1$ with the weight (\ref{weight-R2-1}).
\item \lq\lq cylinder-transmit-adjacent" (green) denotes $\bar \mu_T^{(\Psi_0)}(|\sigma|)$ at $\sigma \in K_2$ labelled by $(i,j,k)=(N/2-1, N/2, 0)$, which is on the central cycle, for S-quantum walk on $\mathcal{K}_1$ with the weight (\ref{weight-R2-1}). 
\item \lq\lq cylinder-hetero" (blue) denotes $\bar \mu_T^{(\Psi_0)}(|\sigma_0|)$ for S-quantum walk on $\mathcal{K}_1$ with the weight (\ref{weight-R2-2}). 
\item \lq\lq tetra-on" (purple) denotes $\bar \mu_T^{(\Psi_0)}(|\sigma_0|)$ for S-quantum walk on $\mathcal{K}_2$ with the weight (\ref{weight-Grover}). Note that $\sigma$ is a part of the tetrahedron in $\mathcal{K}_2$.
\item \lq\lq tetra-off" (sky blue) denotes $\bar \mu_T^{(\Psi_0)}(|\sigma|)$ at $|\sigma| \in K_2$ labelled by $(i,j,k)=(N/2-1, N/2, 0)$ for S-quantum walk on $\mathcal{K}_2$ with the weight (\ref{weight-Grover}). Note that $\sigma$ is on the central cycle but off the tetrahedron in $\mathcal{K}_2$. Comparing with \lq\lq cylinder-transmit-adjacent", localization is not exhibited in this case.
\item \lq\lq Mobius" (yellow) denotes $\bar \mu_T^{(\Psi_0)}(|\sigma_0|)$ for S-quantum walk on $\mathcal{K}_3$ with the weight (\ref{weight-Mobius}). 
\item \lq\lq Mobius-adjacent" (black) denotes $\bar \mu_T^{(\Psi_0)}(|\sigma|)$ at $|\sigma| \in K_2$ labelled by $(i,j,k)=(N/2, N/2, 1)$ for S-quantum walk on $\mathcal{K}_3$ with the weight (\ref{weight-Mobius}). 
\end{itemize}
In all cases which localization occurs, $\bar \mu_T^{(\Psi_0)}(|\sigma|)$ converges to a positive value as $T\to \infty$, while in other cases $\bar \mu_T^{(\Psi_0)}(|\sigma|)$ tends to zero.
\end{figure}

\end{document}